\documentclass[10pt,aps,prl,twocolumn]{revtex4-1}

\usepackage{mathrsfs}
\usepackage{bbm}
\usepackage{amsmath}
\usepackage{amssymb}
\usepackage{graphicx}
\usepackage{mathtools}
\usepackage{amsthm}
\usepackage{MnSymbol}

\makeatletter

\usepackage{hyperref}
\usepackage{color}
\definecolor{mblue}{RGB}{0,50,200}
\hypersetup{
colorlinks,
citecolor=mblue,
linkcolor=mblue,
urlcolor=mblue
}
\allowdisplaybreaks

\newcommand{\mca}{\mathcal}
\newcommand{\msc}{\mathscr}
\newcommand{\mbb}{\mathbb}

\newcommand{\bra}[1]{\left( #1 \right)}
\newcommand{\bras}[1]{\left[ #1 \right]}
\newcommand{\brab}[1]{\left\{ #1 \right\}}

\newcommand{\Tr}[1]{{\rm tr}\left\{#1\right\}}

\newcommand{\Kt}[1]{|#1\rangle}
\newcommand{\Mat}[1]{|#1\rangle\langle #1|}
\newcommand{\dMat}[2]{|#1\rangle\langle #2|}

\newcommand{\davgObs}[3]{\langle #1|#2|#3\rangle}
\newcommand{\sectionprl}[1]{{\em #1}\/.---}

\begin{document}
\title{Finite-Time Quantum Landauer Principle and Quantum Coherence}

\author{Tan Van Vu}
\email{tanvu@rk.phys.keio.ac.jp}

\affiliation{Department of Physics, Keio University, 3-14-1 Hiyoshi, Kohoku-ku, Yokohama 223-8522, Japan}

\author{Keiji Saito}
\email{saitoh@rk.phys.keio.ac.jp}

\affiliation{Department of Physics, Keio University, 3-14-1 Hiyoshi, Kohoku-ku, Yokohama 223-8522, Japan}

\date{\today}

\begin{abstract}
The Landauer principle states that any logically irreversible information processing must be accompanied by dissipation into the environment.
In this study, we investigate the heat dissipation associated with finite-time information erasure and the effect of quantum coherence in such processes.
By considering a scenario wherein information is encoded in an open quantum system whose dynamics are described by the Markovian Lindblad equation, we show that the dissipated heat is lower-bounded by the conventional Landauer cost, as well as a correction term inversely proportional to the operational time.
To clarify the relation between quantum coherence and dissipation, we derive a lower bound for heat dissipation in terms of quantum coherence.
This bound quantitatively implies that the creation of quantum coherence in the energy eigenbasis during the erasure process inevitably leads to additional heat costs.
The obtained bounds hold for arbitrary operational time and control protocol.
By following an optimal control theory, we numerically present an optimal protocol and illustrate our findings by using a single-qubit system.
\end{abstract}

\maketitle

\sectionprl{Introduction}Any irreversible information processing unavoidably incurs a thermodynamic cost.
This fundamental relationship between information and thermodynamics is embodied in the Landauer principle \cite{Landauer.1961.JRD}.
The principle states that the amount of heat dissipation $Q$ required to erase information is lower-bounded by the entropy change $\Delta S$ of information-bearing degrees of freedom, $\beta Q\ge\Delta S$.
Here, $\beta$ is the inverse temperature of the environment.
This inequality -- referred to as the Landauer bound or limit -- lays a foundation for the thermodynamics of information \cite{Sagawa.2012.PTP,Parrondo.2015.NP,Goold.2016.JPA} and computation \cite{Bennett.1982.IJTP,Wolpert.2019.JPA} and provides a resolution to the paradox of Maxwell's demon \cite{Maruyama.2009.RMP}.
The attainability of the Landauer bound in the slow quasistatic limit has been experimentally verified in various systems \cite{Brut.2012.N,Jun.2014.PRL,Yan.2018.PRL,Hong.2016.SA,Saira.2020.PRR,Dago.2021.PRL}.
Nonetheless, modern computing requires fast memory erasure, which generally comes with a thermodynamic cost far beyond the Landauer limit.
Thus, improving our understanding of heat dissipation in finite-time information processing is relevant to the development of efficient computing devices.

In recent years, the Landauer principle has been extensively studied in the framework of stochastic thermodynamics from the classical to the quantum regime \cite{Sagawa.2009.PRL,Esposito.2011.EPL,Hilt.2011.PRE,Deffner.2013.PRX,Reeb.2014.NJP,Browne.2014.PRL,Lorenzo.2015.PRL,Goold.2015.PRL,Alhambra.2017.PRA,Campbell.2017.PRA,Guarnieri.2017.NJP,Boyd.2018.PRX,Klaers.2019.PRL,Shiraishi.2019.PRL,Dechant.2019.arxiv,Buscemi.2020.PRA,Timpanaro.2020.PRL,Wolpert.2020.arxiv,Vu.2021.PRL2,Riechers.2021.PRA}. 
One central issue is investigating how much heat needs to be dissipated to erase information in far-from-equilibrium situations. 
When the erasure fidelity is predetermined, a tradeoff between dissipation and operation time apparently occurs, as indicated in slow driving protocols \cite{Zulkowski.2014.PRE,Zulkowski.2015.PRE,Zulkowski.2015.PRE2}. 
Recently, Proesmans and coworkers have derived a tradeoff relation for arbitrary driving speed in classical bits modeled by a double-well potential \cite{Proesmans.2020.PRL,Proesmans.2020.PRE}. 
They showed that the minimum dissipation for erasing a classical bit is bounded from below by the Landauer cost, as well as a term inversely proportional to the operational time, which is somewhat reminiscent of a speed limit \cite{Shiraishi.2018.PRL,Funo.2019.NJP,Vo.2020.PRE,Vu.2021.PRL}. 
In the short-time limit, the finite-time correction term is dominant over the Landauer cost. 
From this development, quantum extensions relevant to information erasure in qubits are strongly desired.

In addition, in the quantum regime, quantum coherence is one of the crucial aspects. Recently, the role of quantum coherence has been intensively discussed in the context of finite-time thermodynamics \cite{Horodecki.2013.NC,Uzdin.2015.PRX,Lostaglio.2015.PRX,Korzekwa.2016.NJP,Francica.2019.PRE,Santos.2019.npjQI,Francica.2020.PRL,Tajima.2021.PRL}. In the case of conventional heat engines driven by heat baths with different temperatures, quantum coherence hinders the thermodynamic performance in the linear response regime with respect to small driving amplitudes \cite{Brandner.2017.PRL} or small driving speeds \cite{Brandner.2020.PRL}. On the other hand, power outputs can be enhanced by coherence in several far-from-equilibrium models \cite{Scully.2011.PNAS,Uzdin.2015.PRX,Watanabe.2017.PRL,Menczel.2020.PRA}. Thus, the role of quantum coherence in nonequilibrium thermodynamics is highly elusive. Concerning the Landauer principle, Miller and coworkers have examined the slow-driving case and have found that quantum coherence generates additional dissipation \cite{Miller.2020.PRL.QLP}. This property is consistent with the aforementioned slow-driving case in heat engines.
Hence, as the next step toward a complete understanding of the quantum Landauer principle, determining the effect of an arbitrary driving speed on the relation between quantum coherence and dissipation is clearly important.

From these two backgrounds, in this paper, we address the following questions: (i) What is the effect of a finite-time protocol with a finite erasure error on heat dissipation? (ii) What is the role of quantum coherence in heat dissipation for an arbitrary driving speed? These two are clearly fundamental for the in-depth understanding of the quantum Landauer principle.
To answer these questions, we consider a dynamical class described by the Markovian Lindblad equations, which guarantee thermodynamically consistent dynamics.
We then rigorously provide quantitative answers to these questions. The result for the first question is presented below in Ineq.~\eqref{eq:main.result.1}, which is identified as a quantum extension of the classical version obtained in Ref.~\cite{Proesmans.2020.PRL}. 
For the second question, we show the inevitable heat dissipation caused by quantum coherence as below in Ineq.~\eqref{eq:main.result.2}, which explicitly shows that quantum coherence in the qudit always induces additional heat dissipation compared to classical protocols. 
The obtained results hold for an arbitrary control protocol and the driving speed. 
By using a single-qubit system, we present an optimal control protocol obtained numerically by optimizing the dissipation and erasure fidelity, which supports our findings.

\sectionprl{Erasure setup and first main result}We consider an information erasure process realized by using a controllable open qudit system with an arbitrary dimension $d$ and an infinite heat bath at the inverse temperature $\beta=1/k_{\rm B}T$. 
The former and latter are referred to as information-bearing and non-information-bearing degrees of freedom, respectively. 
The information content we want to erase is encoded in the density matrix of the system, which is typically set to a maximally mixed state $\varrho_0=\mbb{I}/d$. This mixed state is regarded as an ensemble of many initial pure states that are subject to reset.
The maximally mixed state is sufficient to understand the average dissipated heat of the erasure process for all initial pure states, and if an erasure protocol can reliably reset the qudit from the maximally mixed state, then it does so for an arbitrary pure state (see Supplemental Material (SM) \cite{Supp.PhysRev} for details).
The information is erased in a finite time $\tau$ by varying the control Hamiltonian $H_t$. 
Here, we focus on a dynamical class in which the system is always weakly attached to the heat bath, and its dynamics are described by the Lindblad master equation \cite{Lindblad.1976.CMP,Gorini.1976.JMP}. Let $\varrho_t$ be the density matrix of the system at time $t$; then, the dynamics are described as 
\begin{equation}\label{eq:Lindblad.eq}
\begin{aligned}
\dot{\varrho}_t & =\mca{L}_t(\varrho_t)\, ,\\
\mca{L}_t(\varrho) &\coloneqq -i\hbar^{-1}[H_t,\varrho]+\sum_k\mca{D}[L_k(t)]\varrho,
\end{aligned}
\end{equation}
with the dissipator given by $\mca{D}[L]\varrho\coloneqq L\varrho L^\dagger-\brab{L^\dagger L,\varrho}/2$. 
Here, $L_k(t)$ are time-dependent jump operators that account for transitions between different energy eigenstates. 
The dot indicates the time derivative, and $[\circ,\star]$ and $\{\circ,\star\}$ denote the commutator and anticommutator of the two operators, respectively.
Assume that the generator $\mca{L}_t$ obeys the quantum detailed balance with respect to the Hamiltonian $H_t$ at all times \cite{Alicki.1976.RMP}.
This condition ensures that the thermal state $\pi_t:=e^{-\beta H_t}/\Tr{e^{-\beta H_t}}$ is the instantaneous stationary state of the Lindblad equation, $\mca{L}_t(\pi_t)=0$.
Hereinafter, both the Planck constant and the Boltzmann constant are set to unity: $\hbar=k_{\rm B}=1$.

Resetting the system state to a specific state results in a change in the system entropy $\Delta S$, which is quantified by the von Neumann entropy. 
The entropy decrease in the information-bearing degrees of freedom is compensated by the amount of heat transferred to the environment $Q$. 
The change in the system entropy and the average dissipated heat are, respectively, written as follows \cite{Alicki.1979.JPA}:
\begin{equation}
\begin{aligned}
\Delta S &\coloneqq -\Tr{\varrho_0\ln\varrho_0}+\Tr{\varrho_\tau\ln\varrho_\tau},\\
Q &\coloneqq -\int_0^\tau\Tr{H_t\dot{\varrho}_t}dt.
\end{aligned}
\end{equation}
The Landauer bound can be immediately derived from the second law of thermodynamics, which states that the irreversible entropy production $\Sigma_\tau$ during period $\tau$ is always nonnegative \cite{Spohn.1978.JMP}, $\Sigma_\tau=-\Delta S+\beta Q\ge 0$.

Under the given setup, we now show the first main result, leaving the details of the proof in the SM \cite{Supp.PhysRev}:
\begin{equation}\label{eq:main.result.1}
  \beta Q\ge\Delta S+\frac{\|\varrho_0-\varrho_\tau\|_1^2}{2\tau\overline{\gamma}_\tau}\ge \Delta S+\frac{[2(1-1/d)-\epsilon]^2}{2\tau\overline{\gamma}_\tau}.
\end{equation}
Here, $\|...\|_1$ is the trace norm, $\epsilon\coloneqq\|\varrho_\tau-\Mat{0}\|_1$ quantifies the distance error to the ground state, and $\overline{\gamma}_\tau$ is the time-averaged dynamical activity that characterizes the timescale of the thermal relaxation, defined as $\overline{\gamma}_\tau\coloneqq\tau^{-1}\int_0^\tau\sum_k\Tr{L_k(t)\varrho_tL_k(t)^\dagger}dt$ \cite{Shiraishi.2018.PRL,Maes.2020.PR}. 
The obtained bound implies that the cost of erasing the information within a finite time is at least the Landauer cost plus a distance term proportional to $\tau^{-1}$. 
The result holds for {\it arbitrary} operational time, driving speed, and the final state, which can be far from the ground state.

Physically important aspects of this relation are now in order. 
First, suppose that we use a protocol in which the stored information is fully erased -- that is, $\varrho_\tau$ is equal to $\Mat{0}$.
Then, the bound Ineq.~\eqref{eq:main.result.1} becomes
\begin{equation}
\beta Q\ge \ln d+\frac{2(1-1/d)^2}{\tau\overline{\gamma}_\tau}.\label{eq:full.erase.bound}
\end{equation}
For a fast erasure $\tau\bar{\gamma}_{\tau}\ll 1$, the second term in the lower bound becomes dominant. 
On the other hand, the lower bound reduces exactly to the Landauer cost in the slow-erasure limit $\tau\bar{\gamma}_{\tau}\gg 1$. 
The inequality \eqref{eq:full.erase.bound} is the quantum extension of the classical finite-time Landauer bound derived in Ref.~\cite{Proesmans.2020.PRL} (see Ineq.~(9) therein).
Next, consider the case of imperfect information erasure; that is, $\varrho_\tau$ is not equal to $\Mat{0}$. 
In this case, we can discuss how quantum coherence remaining at the final time is related to dissipation, thus revealing an intrinsic difference between classical and quantum protocols. 
To this end, let $\Lambda(\circ)\coloneqq\sum_n\Pi_n(\circ)\Pi_n$ be the dephasing map in the energy eigenbasis of the ending Hamiltonian $H_\tau$, in which $\Pi_n$ is the projection operator onto the $n$th eigenspace. 
Because the trace norm is contractive under completely positive trace-preserving maps, we can decompose $\|\varrho_0-\varrho_\tau\|_1^2$ into classical and quantum parts as
\begin{equation}
\|\varrho_0-\varrho_\tau\|_1^2=\|\Lambda(\varrho_0)-\Lambda(\varrho_\tau)\|_1^2+C_{\rm res},
\end{equation}
where $C_{\rm res}\ge 0$ is a quantum term quantifying the residual coherence in the final state with respect to the energy eigenbasis $\{\Kt{n_\tau}\}$.
Analogously, the von Neumann entropy production can also be decomposed as
\begin{equation}
\Delta S=\Delta S_{\rm cl}+C_{\rm rel},
\end{equation}
where $\Delta S_{\rm cl}\coloneqq S(\Lambda(\varrho_0))-S(\Lambda(\varrho_\tau))$ is the classical entropy change in terms of the population distribution and $C_{\rm rel}\coloneqq S(\Lambda(\varrho_\tau))-S(\varrho_\tau)\ge 0$ is the relative entropy of coherence in the final state \cite{Baumgratz.2014.PRL}.
Consequently, the lower bound in Ineq.~\eqref{eq:main.result.1} can be separated into classical and quantum terms as
\begin{equation}
\beta Q\ge \underbrace{\Delta S_{\rm cl}+\frac{\|\Lambda(\varrho_0)-\Lambda(\varrho_\tau)\|_1^2}{2\tau\overline{\gamma}_\tau}}_{\text{classical}}+\underbrace{C_{\rm rel}+\frac{C_{\rm res}}{2\tau\overline{\gamma}_\tau}}_{\text{quantum}}.\label{eq:class.quantum.decomp}
\end{equation}
The classical term is a bound that can also be derived by using a classical probabilistic process. 
The above expression indicates that the quantum coherence left in the final state contributes a nonnegative term in the lower bound.

\sectionprl{Second main result}The relation (\ref{eq:class.quantum.decomp}) describes the quantum coherence in the final state. 
However, this is not enough to capture the role of coherence because quantum coherence is generated during the time it takes to reach the final state. 
Therefore, in addition to the above argument, we discuss the general relation between quantum coherence and heat dissipation during a finite time $\tau$. 
To achieve a general result, we assume a general initial density matrix and consider the time evolution obeying the dynamics (\ref{eq:Lindblad.eq}). 
Quantum coherence in the energy eigenbasis is typically generated because of the presence of non-commuting terms in the Hamiltonian. 
For the present aim, using the following intuitive $\ell_1$ norm of coherence \cite{Baumgratz.2014.PRL} conveniently helps quantify the amount of coherence contained in the state $\varrho_t$:
\begin{equation}
C_{\ell_1}(\varrho_t)\coloneqq\sum_{m\neq n}|\davgObs{m_t}{\varrho_t}{n_t}|.
\end{equation}
Here, $\{\Kt{n_t}\}$ is the instantaneous energy eigenbasis of the Hamiltonian $H_t$ \cite{fnt1}. 
Mathematically, $C_{\ell_1}$ is the sum of the absolute values of all the off-diagonal elements with respect to the energy eigenbasis. 
This quantifier was shown to be a suitable measure of coherence \cite{Baumgratz.2014.PRL} and is widely used in literature \cite{Streltsov.2017.RMP}. 
The amount of coherence accumulated during period $\tau$ can thus be defined naturally:
\begin{equation}
\msc{C}_\tau\coloneqq\int_0^\tau C_{\ell_1}(\varrho_t)dt.
\end{equation}
The second main result of this study is that the average dissipated heat in a finite time is lower-bounded by the Landauer cost plus a nonnegative coherence term (see the SM \cite{Supp.PhysRev}):
\begin{equation}
\beta Q\ge \Delta S+\frac{\overline{\gamma}_\tau^{\rm R}\msc{C}_\tau^2}{2\tau}.\label{eq:main.result.2}
\end{equation}
In this equation, $\overline{\gamma}_\tau^{\rm R}$ is the time average of the characteristic relaxation rate, given by $(\overline{\gamma}_\tau^{\rm R})^{-1}=\tau^{-1}\int_0^\tau\bra{\sum_k\|L_k(t)\|_\infty^2}^{-1}dt$ in the $d=2$ case (see the SM \cite{Supp.PhysRev} for the form of $\overline{\gamma}_\tau^{\rm R}$ in the generic case). 
Here, $\|...\|_{\infty}$ is the spectral norm.
We emphasize that the result holds for the {\it arbitrary} operational time and driving speed.
The obtained relation sets a lower bound on dissipation in terms of information-theoretic entropy and quantum coherence. 
The inequality \eqref{eq:main.result.2} implies that the quantum coherence produced during information erasure must be accompanied by additional heat. 
The greater the generation of coherence, the more heat is dissipated. 
The relation \eqref{eq:main.result.2} is valid for the entire driving speed regime, thus covering the slow-driving limit \cite{Miller.2020.PRL.QLP}. 
Combining this with the first main result (\ref{eq:main.result.1}), the Landauer bound can be strengthened as $\beta Q\ge \Delta S+\max\{\|\varrho_0-\varrho_\tau\|_1^2 /\overline{\gamma}_\tau, \overline{\gamma}_\tau^{\rm R}\msc{C}_\tau^2\}/(2\tau)$, when we use the maximally mixed state for the initial state.

\sectionprl{Numerical demonstration with the optimal control theory}We exemplify our findings with a simple model of information erasure by using a single qubit.
Two-level qubit systems are relevant in quantum computation and are commonly used to store memory in measurement-driven engines \cite{Elouard.2017.PRL,Bresque.2021.PRL}.
The qubit can be viewed as a spin-$1/2$ particle weakly coupled to a large bath of bosonic harmonic oscillators \cite{Leggett.1987.RMP}, evolving according to the Hamiltonian
\begin{equation}
H_t=\frac{\epsilon_t}{2}\bras{\cos(\theta_t)\sigma_z+\sin(\theta_t)\sigma_x}
\end{equation}
and two jump operators $L_1(t)=\sqrt{\alpha\epsilon_t(N_t+1)}\dMat{0_t}{1_t}$, $L_2(t)=\sqrt{\alpha\epsilon_tN_t}\dMat{1_t}{0_t}$, in which $\sigma_{x,y,z}$ are the Pauli matrices, $\alpha$ is the coupling strength, $N_t\coloneqq 1/(e^{\beta\epsilon_t}-1)$ is the Planck distribution, and $\epsilon_t$ and $\theta_t$ are time-dependent control parameters.
The quantity $\epsilon_t$ is the energy gap between the instantaneous energy eigenstates, whereas $\theta_t$ controls the relative strength of coherent tunneling to energy bias \cite{Leggett.1987.RMP}.
If $\theta_t$ is fixed at all times, it corresponds to a classical protocol.
Otherwise, quantum coherence in the energy eigenbasis is generated, implying a genuine quantum protocol.

Resetting the qubit to the ground state $\Kt{0}$ with a probability close to $1$ can be achieved with various control protocols.
For example, we can either gradually increase the energy gap $\epsilon_t$ from an initial value $\epsilon_0\approx 0$ to a final value $\beta\epsilon_\tau\gg 1$ while also changing $\theta_t$ \cite{Miller.2020.PRL.QLP}, or we can quench the Hamiltonian at $t=0$ and let the system relax to an equilibrium state close to $\Mat{0}$.
Here we particularly consider the Pareto-optimal protocols \cite{Miettinen.1999}, which optimize two incompatible objectives: the success probability and the average dissipated heat.
Specifically, we solve the optimization problem of minimizing the following multi-objective functional:
\begin{equation}
\mca{F}[\{\epsilon_t,\theta_t\}]\coloneqq \lambda\tau\beta Q-(1-\lambda)F(\varrho_\tau,\Mat{0}).\label{eq:opt.functional}
\end{equation}
Here, $\lambda\in[0,1)$ is a weighting factor, and $F(\varrho,\sigma)=\Tr{\sqrt{\sqrt{\varrho}\sigma\sqrt{\varrho}}}^2$ is the fidelity of the two quantum states $\varrho$ and $\sigma$ \cite{Jozsa.1994.JMO}.
The first and second terms of the functional correspond to the average dissipated heat and erasure fidelity, respectively.
When $\lambda=0$, the reliability of the information erasure takes precedence over dissipation, and the final state is optimized to be as close to the ground state as possible.
The $\lambda>0$ case indicates that the dissipation is also minimized under a given allowable error of the final state.
Because of physical limitations, imposing constraints on the control parameters is reasonable.
Hereinafter, we set the following lower and upper bounds on the parameters $\beta\epsilon_t\in [0.4,10]$ and $\theta_t\in[-\pi,\pi]$.

\begin{figure}[t]
\centering
\includegraphics[width=1.0\linewidth]{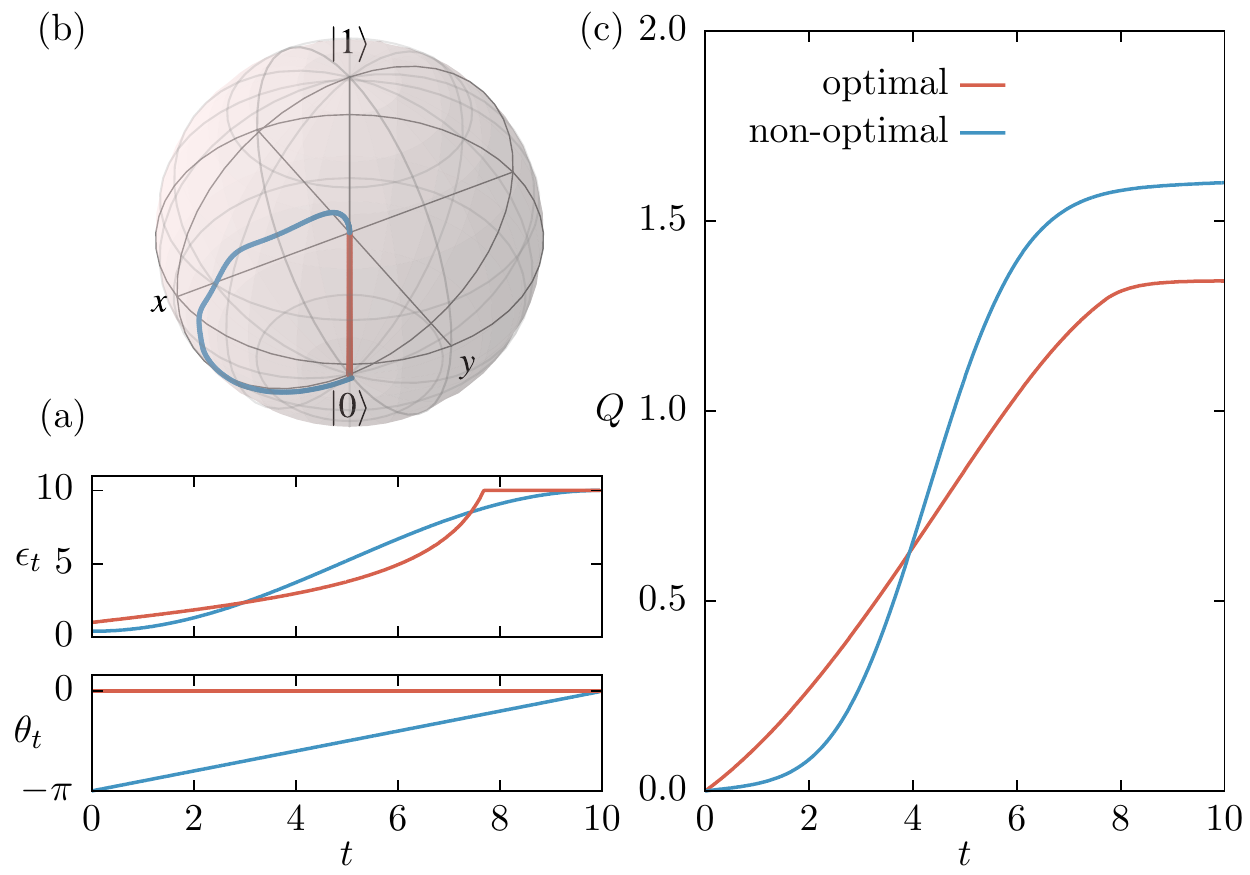}
\protect\caption{Numerical results obtained with the optimal and non-optimal protocols, depicted by the red and blue lines, respectively. (a) Time variations of the control parameters $\epsilon_t$ and $\theta_t$. (b) Geometrical representation of the time evolution of the qubit on the Bloch sphere. (c) Average dissipated heat over time. Other parameters are given by $\alpha=0.2$, $\beta=1$, $\epsilon_0=0.4$, $\epsilon_\tau=10$, and $\tau=10$.}\label{fig:result1}
\end{figure}

Obtaining the analytical solution for the optimization problem in Eq.~\eqref{eq:opt.functional} under the given constraints is a daunting task.
Hence, we numerically solve the optimal protocol by discretizing the protocol and minimizing the functional $\mca{F}$ by using a nonlinear programming method \cite{Supp.PhysRev}.
To demonstrate the effect of quantum coherence on information erasure, we consider $\varrho_0=\mbb{I}/2$.
The initial state thus does not contain any amount of coherence.
We examine two control protocols: the {\it optimal protocol} found through the nonlinear programming method for $\lambda>0$ satisfying $(1-\lambda)/\lambda=10^4$ and the {\it non-optimal protocol} used in Ref.~\cite{Miller.2020.PRL.QLP}, in which $\epsilon_t=\epsilon_0+(\epsilon_\tau-\epsilon_0)\sin(\pi t/2\tau)^2$ and $\theta_t=\pi(t/\tau-1)$.
Both protocols drive the system to the ground state with the same order of error at the final time.
Figure~\ref{fig:result1}(a) shows the time variation of the control parameters $\epsilon_t$ and $\theta_t$ for each protocol.
The energy gap $\epsilon_t$ in each protocol increases gradually in different ways and eventually reaches the same value.
Interestingly, in the optimal protocol, $\theta_t$ is always fixed at $0$; hence, no coherence is created during the period $\tau$.
On the other hand, the non-optimal protocol constantly changes $\theta_t$ and generates quantum coherence.
We express the density matrix in the Bloch representation and plot the time evolution of the Bloch vector in Fig.~\ref{fig:result1}(b).
Notice that the qubit evolves through completely different paths toward the ground state for each protocol.
Figure \ref{fig:result1}(c) shows the average dissipated heat for each protocol at each time.
Note also that the optimal protocol clearly dissipates less heat than the non-optimal protocol, which is consistent with our finding that the generation of coherence incurs additional heat costs.

From an energetic point of view, quantum coherence has been shown to be detrimental to the erasure of information.
To optimize dissipation, the qubit should behave as a classical bit.
Further, we numerically find that this is the case even when dissipation is not minimized -- that is, when $\lambda=0$.
In other words, quenching the Hamiltonian to $H_0=\epsilon_\tau\sigma_z/2$ at time $t=0$ and relaxing the system to equilibrium is the best protocol to bring the qubit as close as possible to the ground state.
We conjecture that the quench protocol, namely, thermal relaxation, is the optimal protocol in terms of erasure reliability.
Note that the conjecture is restricted to the $\varrho_0=\mbb{I}/2$ case because the shortcut-to-equilibration protocol \cite{Dann.2019.PRL} or an optimal protocol \cite{Supp.PhysRev} may outperform the quench protocol for the $\varrho_0\neq\mbb{I}/2$ case.
When the initial state is not described by the maximally mixed state and contains some coherence, we also find that the optimal protocol hardly produces quantum coherence as compared to the non-optimal protocol (see the SM \cite{Supp.PhysRev} for details).
From the analytical and numerical evidence, we can conclude that the creation of quantum coherence should be avoided when erasing information.

\begin{figure}[b]
\centering
\includegraphics[width=1.0\linewidth]{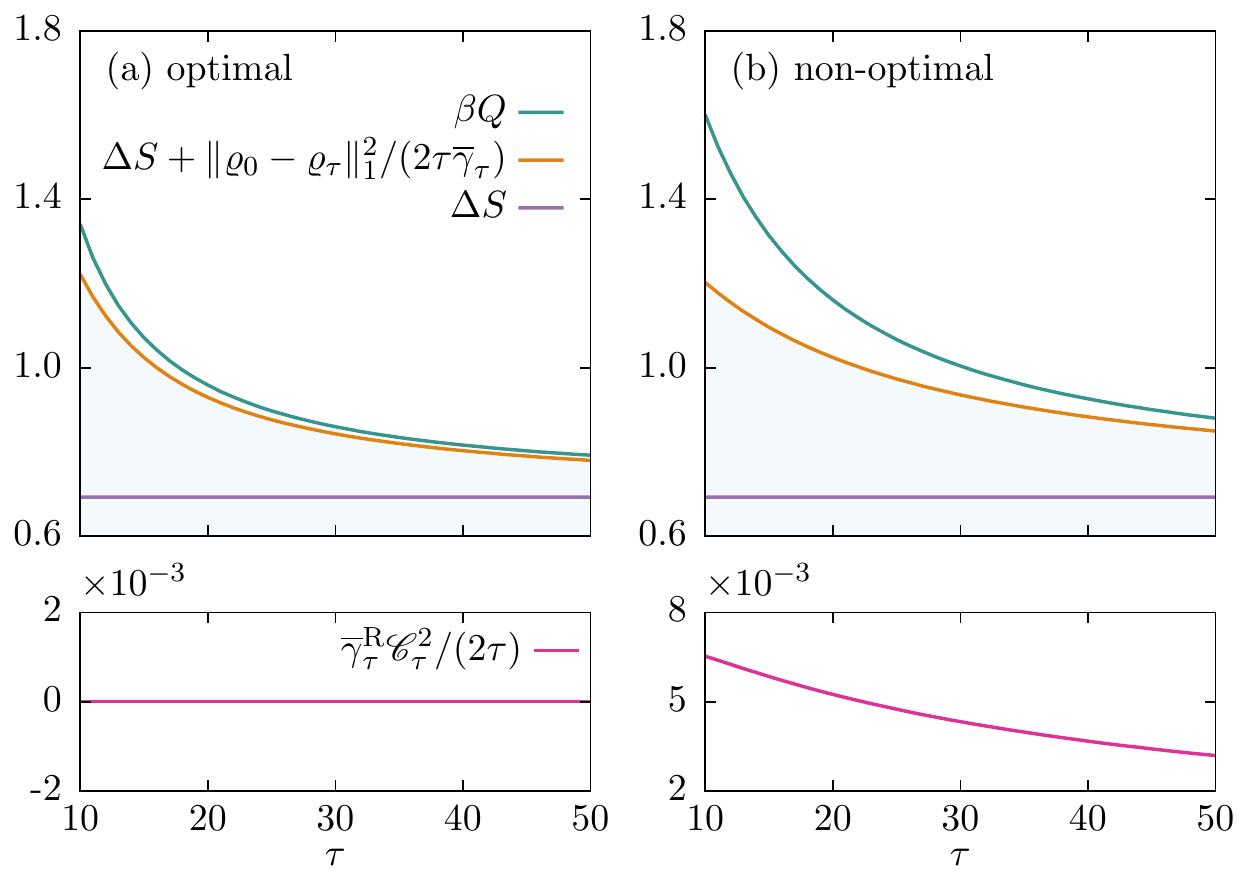}
\protect\caption{Numerical illustration of the bound Ineq.~\eqref{eq:main.result.1} with the (a) optimal and (b) non-optimal protocols. The upper panels show the average dissipated heat $\beta Q$, the derived bound $\Delta S+\|\varrho_0-\varrho_\tau\|_1^2/(2\tau\overline{\gamma}_\tau)$, and the Landauer cost $\Delta S$ for each operational time. The lower panels show the amount of quantum coherence generated in the energy eigenbasis. The operational time $\tau$ is varied, while the other parameters are set to $\alpha=0.2$, $\beta=1$, $\epsilon_0=0.4$, and $\epsilon_\tau=10$.}\label{fig:result2}
\end{figure}

Next, we investigate the performance of the bound Ineq.~\eqref{eq:main.result.1} with optimal and non-optimal protocols.
We vary the operational time and plot the dissipated heat $\beta Q$, the derived bound $\Delta S+\|\varrho_0-\varrho_\tau\|_1^2/(2\tau\overline{\gamma}_\tau)$, the Landauer cost $\Delta S$, and the coherence term $\overline{\gamma}_\tau^{\rm R}\msc{C}_\tau^2/(2\tau)$ as functions of $\tau$ in Fig.~\ref{fig:result2}.
Notice that the dissipated heat is always bounded from below by the bound Ineq.~\eqref{eq:main.result.1} and is far beyond the Landauer cost.
Particularly, in the case of the optimal protocol, the derived bound is tight and asymptotically saturated as $\tau$ increases.
The optimal protocol is also less dissipative than the non-optimal protocol for all operational times.
Simultaneously, the optimal protocol does not create coherence, whereas a positive amount of coherence is generated in the non-optimal protocol.
Regarding the tightness of the derived bounds, Ineq.~\eqref{eq:main.result.1} can be tighter or looser than Ineq.~\eqref{eq:main.result.2}.
If little or no coherence is produced, the former is generally tighter than the latter.
Conversely, when a large amount of coherence is generated -- that is, when $\msc{C}_\tau\gg 1$ -- the latter is stronger than the former (see the SM \cite{Supp.PhysRev} for numerical illustrations).

\sectionprl{Summary}We derived the lower bound on the thermodynamic cost associated with finite-time information erasure for Markovian open quantum dynamics.
The bound is far beyond the Landauer cost for fast control protocols.
We also revealed the relation between quantum coherence and heat dissipation for the entire driving speed regime, stating that the creation of quantum coherence inevitably causes additional heat costs. 
In the context of the Landauer principle, this relation implies that quantum coherence is detrimental to erasing information from an energetic viewpoint.
We confirmed the results with both optimal and non-optimal protocols.
Our findings are not only fundamentally critical but also helpful in establishing a design principle for efficient memory erasure.
The generalization of the results obtained here to other cases, such as the finite-size environments and the non-Markovian regime \cite{Rivas.2020.PRL}, is a future study.

\begin{acknowledgments}
\sectionprl{Acknowledgments}We are grateful to K. Funo and H. Tajima for the fruitful discussion. 
We also thank K. Brandner for telling us about his study on quantum heat engines.
This work was supported by Grants-in-Aid for Scientific Research (JP19H05603 and JP19H05791).
\end{acknowledgments}


\begin{thebibliography}{83}
\bibitem [{fnt()}]{fnt1}%
  \BibitemOpen
  \href@noop {} {}In the case of energy degeneracy, essentially the same result
  can be obtained by using quantum coherence between energy eigenstates with
  different energies \cite{Supp.PhysRev}.\BibitemShut {Stop}%
\end{thebibliography}
\end{document}


\title{Supplemental Material for \\ ``Finite-Time Quantum Landauer Principle and Quantum Coherence''}

\author{Tan Van Vu}

\affiliation{Department of Physics, Keio University, 3-14-1 Hiyoshi, Kohoku-ku, Yokohama 223-8522, Japan}

\author{Keiji Saito}

\affiliation{Department of Physics, Keio University, 3-14-1 Hiyoshi, Kohoku-ku, Yokohama 223-8522, Japan}

\begin{abstract}
This supplemental material describes the details of analytical and numerical calculations presented in the main text. 
The equations and figure numbers are prefixed with S [e.g., Eq.~(S1) or Fig.~S1]. 
Numbers without this prefix [e.g., Eq.~(1) or Fig.~1] refer to items in the main text.

\end{abstract}

\pacs{}
\maketitle

\tableofcontents

\section{Erasure process with the initial maximally mixed state}
Here we discuss some crucial implications of considering the erasure process with the initial maximally mixed state, i.e., $\varrho_0=\mbb{I}/d$.
Suppose that we want to reset an unknown qudit (which can be in an arbitrary pure state) to the ground state $\Kt{0}$.
It is worth noting that the erasure process may be imperfect (i.e., the qudit may not be reset exactly to the ground state) because we are considering finite-time information erasure.
We define the completely positive trace-preserving map that represents the erasure protocol as $\Lambda(\cdot):\varrho_0\mapsto\varrho_\tau$.
In the case of the Lindblad dynamics, this map reads $\Lambda(\varrho_0)={\mca{T}}\exp\bra{\int_0^\tau \mca{L}_tdt}\varrho_0$, where $\mca{T}$ denotes the time-ordering operator.

First, we show that studying the erasure process with the initial maximally mixed state suffices to understand the average heat dissipation required to erase an unknown qudit.
It is evident that $\mbb{E}[\Mat{\psi}]=\mbb{I}/d$, where the average $\mbb{E}[\cdot]$ is taken over all possible pure states $\Kt{\psi}=U\Kt{0}$ with $U$ is a Haar-random unitary operator.
Let $\mca{Q}(\varrho_0)$ be the dissipated heat of the protocol for the initial state $\varrho_0$, then the average dissipated heat over all pure states can be calculated as $\overline{Q}=\mbb{E}[\mca{Q}(\Mat{\psi})]$.
Using the linearity of $\Lambda$ and $\mca{Q}$ [i.e., $\Lambda(\lambda\varrho+(1-\lambda)\sigma)=\lambda\Lambda(\varrho)+(1-\lambda)\Lambda(\sigma)$ and $\mca{Q}(\lambda\varrho+(1-\lambda)\sigma)=\lambda\mca{Q}(\varrho)+(1-\lambda)\mca{Q}(\sigma)$], we obtain
\begin{equation}
\overline{Q}=\mbb{E}[\mca{Q}(\Mat{\psi})]=\mca{Q}(\mbb{E}[\Mat{\psi}])=\mca{Q}(\mbb{I}/d).\label{eq:avg.heat}
\end{equation}
Equation \eqref{eq:avg.heat} implies that the dissipated heat in the case of the initial maximally mixed state is exactly the average dissipated heat of the erasure protocol over all pure states.

Next, we prove that if an erasure protocol erases the maximally mixed state with the error order of $O(\epsilon)$, then the protocol erases an arbitrary pure state with the error order of $O(\epsilon^{1/2})$.
Let $\|A\|_{\rm F}=\sqrt{\Tr{A^\dagger A}}=\sqrt{\sum_{m,n}|a_{mn}|^2}$ be the Frobenius norm of a matrix $A=[a_{mn}]$, the statement is explicitly embodied in the following lemma.
\begin{lemma}
If an erasure protocol satisfies $\|\Lambda(\mbb{I}/d)-\Mat{0}\|_{\rm F}^2\le\epsilon$ (where $\epsilon>0$ is a sufficiently small number), then for an arbitrary pure state $\Kt{\psi}$, the following inequality holds:
\begin{equation}
\|\Lambda(\Mat{\psi})-\Mat{0}\|_{\rm F}^2\le 2\sqrt{d(d-1)\epsilon}\xrightarrow{\epsilon\to 0} 0.
\end{equation}
\end{lemma}
\begin{proof}
Let $\{\chi_n\}_{0\le n\le d^2-1}$ denote an orthonormal basis for the linear space of operators in the complex Hilbert space $\mca{H}^d$.
Specifically, $\chi_{n}$ with $n=(j-1)d+(k-1)~(1\le j,k\le d)$ can be the generalized Gell-Mann matrices, given as follows:
\begin{equation}
\chi_n=\begin{cases}
(E_{k,j}+E_{j,k})/\sqrt{2}, & {\text{if}}~ j<k, \\
i(E_{k,j}-E_{j,k})/\sqrt{2}, & {\text{if}}~ j>k, \\
\sqrt{\frac{2}{j(j+1)}}\bra{\sum_{l=1}^{j}E_{l,l}-jE_{j+1,j+1}}/\sqrt{2}, & {\text{if}}~ j=k<d, \\
\mbb{I}/\sqrt{d}, & {\text{if}}~ j=k=d, \\
\end{cases}
\end{equation}
where $E_{j,k}$ denotes a matrix with $1$ in the $jk$-th entry and $0$ elsewhere.
Note that each $\chi_n$ is a Hermitian matrix, $\Tr{\chi_n}=0$ for all $n\neq d^2-1$, and $\Tr{\chi_n\chi_m}=\delta_{nm}$.
An arbitrary density matrix $\varrho$ can be expressed in the following form:
\begin{equation}
\varrho=\frac{1}{\sqrt{d}}\chi_{d^2-1}+\sum_{n=0}^{d^2-2}x_n\chi_n,
\end{equation}
where $x_n=\Tr{\varrho\chi_n}\in\mbb{R}$.
Since
\begin{equation}
1\ge \Tr{\varrho^2}=\frac{1}{d}+\sum_{n=0}^{d^2-2}x_n^2,
\end{equation}
each density matrix $\varrho$ can be mapped with a point $\mbm{x}(\varrho)\coloneqq[x_0,\dots,x_{d^2-2}]^\top$ in a sphere with the radius $r=\sqrt{1-1/d}$. 
A pure state corresponds to a point on the boundary of the sphere, whereas a mixed state is represented by a point inside the sphere.
In addition, the Frobenius norm of the difference of two density matrices $\varrho$ and $\sigma$ can be written in the terms of their sphere points as follows:
\begin{equation}
\|\varrho-\sigma\|_{\rm F}^2={\sum_{n=0}^{d^2-2}|x_n(\varrho)-x_n(\sigma)|^2}=\|\mbm{x}(\varrho)-\mbm{x}(\sigma)\|_F^2.
\end{equation}

Using this sphere representation, we are going to prove the lemma in the following.
Without loss of generality, we can assume that the ground state $\Kt{0}$ correspond to the point $\mbm{g}\coloneqq\mbm{x}(\Mat{0})=[-r,0,\dots,0]^\top$.
This is always possible by applying an appropriate coordinate transformation $\mbm{x}\to \msf{T}\mbm{x}$, where $\msf{T}^\top\msf{T}=\msf{1}$.
Note that the Frobenius norm is invariant under such unitary transforms.
For each pure state $\Kt{\psi}$, there always exists a mixed state ${\psi'}$ such that $\Mat{\psi}+(d-1){\psi'}=\mbb{I}$.
Indeed, the mixed state ${\psi'}$ can be chosen as $\psi'=(\mbb{I}-\Mat{\psi})/(d-1)$.
Let $\mbm{a}\coloneqq\mbm{x}(\Lambda(\Mat{\psi}))$ and $\mbm{a}'\coloneqq\mbm{x}(\Lambda({\psi'}))$, then the sphere point corresponding to the state $\Lambda(\mbb{I}/d)$ can be expressed in terms of $\mbm{a}$ and $\mbm{a}'$ as
\begin{equation}
\mbm{x}(\Lambda(\mbb{I}/d))=\mbm{x}(\Lambda([\Mat{\psi}+(d-1){\psi'}]/d))=[\mbm{x}(\Lambda(\Mat{\psi}))+(d-1)\mbm{x}(\Lambda({\psi'}))]/d=[\mbm{a}+(d-1)\mbm{a}']/d.
\end{equation}
Therefore, the condition $\|\Lambda(\mbb{I}/d)-\Mat{0}\|_{\rm F}^2\le\epsilon$ is equivalent to
\begin{equation}
\|[\mbm{a}+(d-1)\mbm{a}']/d-\mbm{g}\|_{\rm F}^2\le\epsilon.\label{eq:fnorm.tmp1}
\end{equation}
It suffices to prove that $\|\mbm{a}-\mbm{g}\|_{\rm F}^2\le 2\sqrt{d(d-1)\epsilon}$.
From Eq.~\eqref{eq:fnorm.tmp1}, we have
\begin{equation}
|[a_0+(d-1)a_0']/d-g_0|^2\le \|[\mbm{a}+(d-1)\mbm{a}']/d-\mbm{g}\|_{\rm F}^2\le\epsilon\Rightarrow \frac{r+a_0}{d}+(d-1)\frac{r+a_0'}{d}\le\sqrt{\epsilon}\Rightarrow -r\le a_0\le -r+d\sqrt{\epsilon}.\label{eq:fnorm.tmp2}
\end{equation}
Equation \eqref{eq:fnorm.tmp2} immediately derives that $|a_0+r|\le d\sqrt{\epsilon}$ and $|a_0|\ge r-d\sqrt{\epsilon}\ge 0$.
From the inequality $\sum_{n=0}^{d^2-2}|a_n|^2\le r^2$, the partial sum $\sum_{n=1}^{d^2-2}|a_n|^2$ can be upper-bounded as
\begin{equation}
\sum_{n=1}^{d^2-2}|a_n|^2\le r^2-|a_0|^2\le r^2-(r-d\sqrt{\epsilon})^2=2rd\sqrt{\epsilon}-d^2\epsilon.
\end{equation}
Combining these inequalities, we finally have
\begin{equation}
\|\mbm{a}-\mbm{g}\|_{\rm F}^2=|a_0+r|^2+\sum_{n=1}^{d^2-2}|a_n|^2\le d^2\epsilon+2rd\sqrt{\epsilon}-d^2\epsilon=2rd\sqrt{\epsilon}=2\sqrt{d(d-1)\epsilon},
\end{equation}
which completes the proof.
\end{proof}
Roughly speaking, the above lemma implies that the error of an erasure protocol can be estimated by investigating the case of the initial maximally mixed state. If an erasure protocol erases the maximally mixed state perfectly, then it does so for an arbitrary pure state.

\section{Alternative expression of the irreversible entropy production}

We define $U_t\coloneqq \mca{T}\exp\bra{-i\int_0^t H_sds}$, where $\mca{T}$ denotes the time-ordering operator.
This time-dependent operator satisfies the following differential equation:
\begin{equation}
\dot{U}_t=-iH_tU_t.
\end{equation}
For each time-dependent operator $X_t$, we also define the corresponding operator in the interaction picture, $\tilde{X}_t\coloneqq U_t^\dagger X_tU_t$.

First, we show that the density matrix $\tilde{\varrho}_t$ obeys the following Lindblad master equation:
\begin{equation}
\dot{\tilde{\varrho}}_t=\sum_k\mca{D}[\tilde{L}_k(t)]\tilde{\varrho}_t.\label{eq:Lind.eq.int.pic}
\end{equation}
Taking the time derivative of $\tilde{\varrho}_t$, Eq.~\eqref{eq:Lind.eq.int.pic} can be obtained as
\begin{equation}
\begin{aligned}[b]
\dot{\tilde{\varrho}}_t&=U_t^\dagger i[H_t,\varrho_t]U_t+U_t^\dagger\dot{\varrho}_tU_t\\
&=U_t^\dagger\sum_k\bras{L_k(t)\varrho_tL_k(t)^\dagger-\frac{1}{2}\{L_k(t)^\dagger L_k(t),\varrho_t\}}U_t\\
&=\sum_k\bras{\tilde{L}_k(t)\tilde{\varrho}_t\tilde{L}_k(t)^\dagger-\frac{1}{2}\{\tilde{L}_k(t)^\dagger \tilde{L}_k(t),\tilde{\varrho}_t\}}\\
&=\sum_k\mca{D}[\tilde{L}_k(t)]\tilde{\varrho}_t.
\end{aligned}
\end{equation}
It should be noted that the jump operators in the interaction picture satisfy $[\tilde{L}_k(t),\tilde{H}_t]=\omega_k(t)\tilde{L}_k(t)$, $\tilde{L}_k(t)=e^{\beta\omega_k(t)/2}\tilde{L}_{k'}(t)^\dagger$, and the operator $\tilde{\pi}_t$ is the instantaneous equilibrium state of Eq.~\eqref{eq:Lind.eq.int.pic}, i.e., $\sum_k\mca{D}[\tilde{L}_k(t)]\tilde{\pi}_t=0$.
Here, $\omega_k(t)$ is the energy change associated with the jump operator $L_k(t)$, and $L_{k'}(t)$ denotes the reversed jump operator corresponding to $L_k(t)$.

Next, we show that the Lindblad master equation [Eq.~\eqref{eq:Lind.eq.int.pic}] can be written as
\begin{equation}
\dot{\tilde{\varrho}}_t=\mca{O}_{\tilde{\varrho}_t}(t,-\ln\tilde{\varrho}_t+\ln\tilde{\pi}_t),\label{eq:Lind.sup.ope}
\end{equation}
where the time-dependent super-operator $\mca{O}_\phi(t,\nu)$ is defined by
\begin{equation}
\mca{O}_{\phi}(t,\nu)\coloneqq\frac{1}{2}\sum_{k}e^{-\beta\omega_k(t)/2}[\tilde{L}_k(t),\sop{\phi}_{\beta\omega_k(t)}([\tilde{L}_k(t)^\dagger,\nu])],
\end{equation}
and the tilted operator is defined as $\sop{\phi}_\theta(X)\coloneqq e^{-\theta/2}\int_0^1e^{s\theta}\phi^sX\phi^{1-s}ds$.
In other words, we need only show that
\begin{equation}
\mca{O}_{\tilde{\varrho}_t}(t,-\ln\tilde{\varrho}_t+\ln\tilde{\pi}_t)=\sum_k\mca{D}[\tilde{L}_k(t)]\tilde{\varrho}_t.\label{eq:operator.O}
\end{equation}
To this end, we note that $\sop{\phi}_{\theta}([X,\ln\phi]-\theta X)=e^{-\theta/2}X\phi-e^{\theta/2}\phi X$ \cite{Maas.2011.JFA,Vu.2021.PRL} for an arbitrary operator $X$ and $\theta\in\mbb{R}$.
Using the relation $[\tilde{L}_k(t)^\dagger,\tilde{H}_t]=-\omega_k(t)\tilde{L}_k(t)^\dagger$, we immediately obtain
\begin{equation}
\begin{aligned}[b]
\sop{\tilde{\varrho}_t}_{\beta\omega_k(t)}([\tilde{L}_k(t)^\dagger,-\ln\tilde{\varrho}_t+\ln\tilde{\pi}_t])&=\sop{\tilde{\varrho}_t}_{\beta\omega_k(t)}([\tilde{L}_k(t)^\dagger,-\ln\tilde{\varrho}_t-\beta\tilde{H}_t])\\
&=-\sop{\tilde{\varrho}_t}_{\beta\omega_k(t)}([\tilde{L}_k(t)^\dagger,\ln\tilde{\varrho}_t]+\beta[\tilde{L}_k(t)^\dagger,\tilde{H}_t])\\
&=-\sop{\tilde{\varrho}_t}_{\beta\omega_k(t)}([\tilde{L}_k(t)^\dagger,\ln\tilde{\varrho}_t]-\beta\omega_k(t)\tilde{L}_k(t)^\dagger)\\
&=e^{\beta\omega_k(t)/2}\tilde{\varrho}_t\tilde{L}_k(t)^\dagger-e^{-\beta\omega_k(t)/2}\tilde{L}_k(t)^\dagger\tilde{\varrho}_t.
\end{aligned}
\end{equation}
Then, we can verify Eq.~\eqref{eq:operator.O} as follows:
\begin{equation}
\begin{aligned}[b]
\mca{O}_{\tilde{\varrho}_t}(t,-\ln\tilde{\varrho}_t+\ln\tilde{\pi}_t)&=\frac{1}{2}\sum_{k}e^{-\beta\omega_k(t)/2}[\tilde{L}_k(t),\sop{\tilde{\varrho}_t}_{\beta\omega_k(t)}([\tilde{L}_k(t)^\dagger,-\ln\tilde{\varrho}_t+\ln\tilde{\pi}_t])]\\
&=\frac{1}{2}\sum_{k}e^{-\beta\omega_k(t)/2}[\tilde{L}_k(t),e^{\beta\omega_k(t)/2}\tilde{\varrho}_t\tilde{L}_k(t)^\dagger-e^{-\beta\omega_k(t)/2}\tilde{L}_k(t)^\dagger\tilde{\varrho}_t]\\
&=\frac{1}{2}\sum_{k}\bras{-e^{-\beta\omega_k(t)}\tilde{L}_k(t)\tilde{L}_k(t)^\dagger\tilde{\varrho}_t+\tilde{L}_k(t)\tilde{\varrho}_t\tilde{L}_k(t)^\dagger+e^{-\beta\omega_k(t)}\tilde{L}_k(t)^\dagger\tilde{\varrho}_t\tilde{L}_k(t)-\tilde{\varrho}_t\tilde{L}_k(t)^\dagger\tilde{L}_k(t)}\\
&=\frac{1}{2}\sum_{k}\brab{\bras{\tilde{L}_k(t)\tilde{\varrho}_t\tilde{L}_k(t)^\dagger-\tilde{\varrho}_t\tilde{L}_k(t)^\dagger\tilde{L}_k(t)}+\bras{\tilde{L}_{k'}(t)\tilde{\varrho}_t\tilde{L}_{k'}(t)^\dagger-\tilde{L}_{k'}(t)^\dagger\tilde{L}_{k'}(t)\tilde{\varrho}_t}}\\
&=\sum_{k}\bras{\tilde{L}_k(t)\tilde{\varrho}_t\tilde{L}_k(t)^\dagger-\frac{1}{2}\brab{\tilde{L}_k(t)^\dagger\tilde{L}_k(t),\tilde{\varrho}_t}}\\
&=\sum_k\mca{D}[\tilde{L}_k(t)]\tilde{\varrho}_t.
\end{aligned}
\end{equation}
Hereinafter, we define the inner product $\avg{\cdot,\cdot}$ as $\avg{X,Y}\coloneqq\Tr{\Xdg Y}$ for arbitrary operators $X$ and $Y$.
Then, the irreversible entropy production can be calculated as
\begin{equation}
\begin{aligned}[b]
\Sigma_\tau&=\int_0^\tau\bras{-\Tr{\dot{\varrho}_t\ln\varrho_t}-\beta\Tr{H_t\dot{\rho}_t}}dt\\
&=\int_0^\tau\avg{-\ln\varrho_t+\ln\pi_t,\dot{\varrho}_t}dt\\
&=\int_0^\tau\avg{U_t^\dagger(-\ln\varrho_t+\ln\pi_t)U_t,U_t^\dagger\dot{\varrho}_tU_t}dt\\
&=\int_0^\tau\avg{U_t^\dagger(-\ln\varrho_t+\ln\pi_t)U_t,\dot{\tilde{\varrho}}_t-iU_t^\dagger[H_t,\varrho_t]U_t}dt\\
&=\int_0^\tau\avg{U_t^\dagger(-\ln\varrho_t+\ln\pi_t)U_t,\dot{\tilde{\varrho}}_t}dt\\
&=\int_0^\tau\avg{-\ln\tilde{\varrho}_t+\ln\tilde{\pi}_t,\dot{\tilde{\varrho}}_t}dt\\
&=\int_0^\tau\avg{-\ln\tilde{\varrho}_t+\ln\tilde{\pi}_t,\mca{O}_{\tilde{\varrho}_t}(t,-\ln\tilde{\varrho}_t+\ln\tilde{\pi}_t)}dt.
\end{aligned}
\end{equation}

\section{Lower bound in terms of the trace norm [Ineq.~(\DistanceBound) in the main text]}
\subsection{Derivation of the bound}

Here we provide a detailed derivation of the lower bound in terms of the trace norm in the main text.
Before going into details of the derivation, let us state some useful propositions that will be used during the derivation.

\begin{proposition}\label{prop:cs.ine}
Given a density matrix $\phi$ and $t\ge 0$, the following inequality holds for arbitrary operators $\xi$ and $\nu$,
\begin{equation}
|\avg{\xi,\mca{O}_\phi(t,\nu)}|^2\le \avg{\xi,\mca{O}_\phi(t,\xi)}\avg{\nu,\mca{O}_\phi(t,\nu)}.\label{eq:Cauchy.ine}
\end{equation}
\end{proposition}
\begin{proof}
For any density matrix $\phi=\sum_n\phi_n\Mat{\phi_n}$, where $\sum_n\phi_n=1$ and $\{\Kt{\phi_n}\}$ are orthonormal eigenvectors, we can express the tilted operator as
\begin{equation}
\sop{\phi}_\theta(X)=e^{-\theta/2}\int_0^1e^{s\theta}\phi^sX\phi^{1-s}ds=\sum_{n,m}\Phi(e^{\theta/2}\phi_n,e^{-\theta/2}\phi_m)\Br{\phi_n}X\Kt{\phi_m}\Kt{\phi_n}\Br{\phi_m}.\label{eq:tilted.operator}
\end{equation}
Here, $\Phi(x,y)$ is the logarithmic mean of two positive numbers $x$ and $y$, given by $\Phi(x,y)=(x-y)/[\ln(x)-\ln(y)]$ for $x\neq y$ and $\Phi(x,x)=x$.

First, we show that the inner product $\avg{\cdot,\mca{O}_\phi(t,\cdot)}$ satisfies the conjugate-symmetry condition $\avg{\xi,\mca{O}_\phi(t,\nu)}=\avg{\nu,\mca{O}_\phi(t,\xi)}^*$ for arbitrary operators $\nu$ and $\xi$.
Here, $*$ denotes the complex conjugate.
For an arbitrary operator $X$ and a real number $\theta\in\mbb{R}$, we have
\begin{equation}\label{eq:cs.ine1}
\begin{aligned}[b]
\avg{\xi,[X,\sop{\phi}_{\theta}([\Xdg,\nu])]}&=\Tr{\xi^\dagger[X,\sop{\phi}_{\theta}([\Xdg,\nu])]}\\
&=\Tr{[\xi^\dagger,X]\sop{\phi}_{\theta}([\Xdg,\nu])}\\
&=\sum_{n,m}\Phi(e^{\theta/2}\phi_n,e^{-\theta/2}\phi_m)\Br{\phi_n}[\Xdg,\nu]\Kt{\phi_m}\Br{\phi_m}[\xi^\dagger,X]\Kt{\phi_n},
\end{aligned}
\end{equation}
where we have used Eq.~\eqref{eq:tilted.operator} in Eq.~\eqref{eq:cs.ine1}.
Swapping $\xi$ and $\nu$, we obtain
\begin{equation}\label{eq:cs.ine2}
\begin{aligned}[b]
\avg{\nu,[X,\sop{\phi}_{\theta}([\Xdg,\xi])]}^*&=\sum_{n,m}\Phi(e^{\theta/2}\phi_n,e^{-\theta/2}\phi_m)\Br{\phi_n}[\Xdg,\xi]\Kt{\phi_m}^*\Br{\phi_m}[\nu^\dagger,X]\Kt{\phi_n}^*\\
&=\sum_{n,m}\Phi(e^{\theta/2}\phi_n,e^{-\theta/2}\phi_m)\Br{\phi_m}[\xi^\dagger,X]\Kt{\phi_n}\Br{\phi_n}[\Xdg,\nu]\Kt{\phi_m}\\
&=\avg{\xi,[X,\sop{\phi}_{\theta}([\Xdg,\nu])]}.
\end{aligned}
\end{equation}
As $\mca{O}_\phi(t,\nu)=\frac{1}{2}\sum_{k}e^{-\beta\omega_k(t)/2}[\tilde{L}_k(t),\sop{\phi}_{\beta\omega_k(t)}([\tilde{L}_k(t)^\dagger,\nu])]$, Eq.~\eqref{eq:cs.ine2} implies that
\begin{equation}
\avg{\nu,\mca{O}_\phi(t,\xi)}^*=\avg{\xi,\mca{O}_\phi(t,\nu)}.\label{eq:consym1}
\end{equation}

Next, we prove that the inner product $\avg{\cdot,\mca{O}_\phi(t,\cdot)}$ is positive-definite, i.e., $\avg{\nu,\mca{O}_\phi(t,\nu)}\ge 0$ for an arbitrary operator $\nu$.
Indeed, from Eq.~\eqref{eq:cs.ine1}, we can show that
\begin{equation}
\avg{\nu,[X,\sop{\phi}_{\theta}([\Xdg,\nu])]}=\sum_{n,m}\Phi(e^{\theta/2}\phi_n,e^{-\theta/2}\phi_m)|\Br{\phi_n}[\Xdg,\nu]\Kt{\phi_m}|^2\ge 0.
\end{equation}
Therefore, $\avg{\nu,\mca{O}_\phi(t,\nu)}\ge 0$.

Given the above properties, we can conclude that the binary operation $\langle\!\langle{\xi,\nu}\rangle\!\rangle\coloneqq\avg{\xi,\mca{O}_\phi(t,\nu)}$ is an inner product between two linear operators since it satisfies the linearity, conjugate-symmetry, and positive-definite conditions. Consequently, the Cauchy--Schwarz inequality can be applied for any two operators $\xi$ and $\nu$,
\begin{equation}
|\langle\!\langle{\xi,\nu}\rangle\!\rangle|^2\le\langle\!\langle{\xi,\xi}\rangle\!\rangle\langle\!\langle{\nu,\nu}\rangle\!\rangle,
\end{equation}
from which the inequality Eq.~\eqref{eq:Cauchy.ine} is immediately proved.

\end{proof}

\begin{proposition}\label{prop:sop.bound}
For an arbitrary operator $Y$, real number $\theta$, and density matrix $\phi$, the following inequality holds,
\begin{equation}
\avg{Y,\sop{\phi}_\theta(Y)}\le\frac{1}{2}\bra{e^{\theta/2}\Tr{\phi YY^\dagger}+e^{-\theta/2}\Tr{\phi Y^\dagger Y}}\le\frac{1}{2}(e^{\theta/2}+e^{-\theta/2})\|Y\|_\infty^2,
\end{equation}
where $\|Y\|_\infty$ denotes the spectral norm of the operator $Y$.
\end{proposition}
\begin{proof}
Using Eq.~\eqref{eq:tilted.operator}, we have
\begin{equation}
\avg{Y,\sop{\phi}_\theta(Y)}=\sum_{n,m}\Phi(e^{\theta/2}\phi_n,e^{-\theta/2}\phi_m)|\Br{\phi_n}Y\Kt{\phi_m}|^2.
\end{equation}
Applying the inequality $\Phi(x,y)\le(x+y)/2$ and the relation $\sum_{n}\Mat{\phi_n}=\mbb{I}$, we obtain
\begin{equation}\label{eq:bound.tilted.op}
\begin{aligned}[b]
\avg{Y,\sop{\phi}_\theta(Y)}&\le\frac{1}{2}\sum_{n,m}\bra{e^{\theta/2}\phi_n+e^{-\theta/2}\phi_m}|\Br{\phi_n}Y\Kt{\phi_m}|^2\\
&=\frac{1}{2}\sum_{n,m}e^{\theta/2}\phi_n\Br{\phi_n}Y\Kt{\phi_m}\Br{\phi_m}Y^\dagger\Kt{\phi_n}+\frac{1}{2}\sum_{m,n}e^{-\theta/2}\phi_m\Br{\phi_m}Y^\dagger\Kt{\phi_n}\Br{\phi_n}Y\Kt{\phi_m}\\
&=\frac{1}{2}\sum_{n}e^{\theta/2}\phi_n\Br{\phi_n}YY^\dagger\Kt{\phi_n}+\frac{1}{2}\sum_{m}e^{-\theta/2}\phi_m\Br{\phi_m}Y^\dagger Y\Kt{\phi_m}\\
&=\frac{1}{2}\bra{e^{\theta/2}\Tr{\phi YY^\dagger}+e^{-\theta/2}\Tr{\phi Y^\dagger Y}}\\
&\le\frac{1}{2}(e^{\theta/2}+e^{-\theta/2})\|Y\|_\infty^2.
\end{aligned}
\end{equation}
Here we have used $\Tr{\phi YY^\dagger}\le\|Y\|_\infty^2$ and $\Tr{\phi Y^\dagger Y}\le\|Y\|_\infty^2$ in the last inequality.
\end{proof}

\begin{proposition}\label{prop:mat.ine}
For any density matrix $\varrho$ and self-adjoint operator $X$ satisfying $X^2=\mbb{I}$ and $[X,\varrho]=0$, the following inequality always holds for any operator $L$,
\begin{equation}\label{eq:mat.ine1}
\Tr{\varrho[L,X][L,X]^\dagger}\le 4\Tr{\varrho L\Ldg }.
\end{equation}
\end{proposition}
\begin{proof}
Since $AA^\dagger$ is positive semi-definite for any operator $A$, we thus have
\begin{equation}
(LX+XL)(LX+XL)^\dagger\succeq 0\Rightarrow -LX\Ldg \Xdg-XL\Xdg\Ldg\preceq LX\Xdg\Ldg+XL\Ldg\Xdg,
\end{equation}
where $A\preceq B$ means that the matrix $B-A$ is positive semi-definite.
Therefore,
\begin{equation}\label{eq:mat.ine2}
\begin{aligned}[b]
[L,X][L,X]^\dagger &=(LX-XL)(\Xdg \Ldg-\Ldg \Xdg)\\
&=LX\Xdg \Ldg+XL\Ldg \Xdg-LX\Ldg \Xdg-XL\Xdg \Ldg\\
&\preceq 2(LX\Xdg \Ldg+XL\Ldg \Xdg).
\end{aligned}
\end{equation}
Since $\Tr{AB}\ge 0$ for any self-adjoint operators $A,B\succeq 0$, Eq.~\eqref{eq:mat.ine1} can be proved from Eq.~\eqref{eq:mat.ine2} as
\begin{equation}
\begin{aligned}[b]
\Tr{\varrho[L,X][L,X]^\dagger}\le 2\Tr{\varrho(LX\Xdg \Ldg+XL\Ldg \Xdg)}=4\Tr{\varrho LL^\dagger}.
\end{aligned}
\end{equation}
\end{proof}

We are now ready to prove the Ineq.~(\DistanceBound) in the main text.
Let $\varrho_t=\sum_n\mu_n(t)\Mat{\mu_n(t)}$ be the spectral decomposition of the density matrix $\varrho_t$, then the corresponding density matrix in the interaction picture reads $\tilde{\varrho}_t=\sum_n\mu_n(t)\Mat{\tilde{\mu}_n(t)}$.
Since $\varrho_0=\mbb{I}/d=\sum_n\mu_n(0)\Mat{\mu_n(\tau)}$, the trace norm can be written as $\|\varrho_0-\varrho_\tau\|_1=\sum_n|\mu_n(0)-\mu_n(\tau)|$.
We define $\tilde{X}_t\coloneqq\sum_n{\rm sign}[\mu_n(\tau)-\mu_n(0)]\Mat{\tilde{\mu}_n(t)}$, where ${\rm sign}(x)=1$ for $x\ge 0$ and ${\rm sign}(x)=-1$ otherwise.
It can be verified that $\tilde{X}_t$ is self-adjoint and satisfies $\tilde{X}_t^2=\mbb{I}$ and $[\tilde{X}_t,\tilde{\varrho}_t]=0$.
Noting that $\Tr{\tilde{X}_t\dot{\tilde{\varrho}}_t}=\sum_n{\rm sign}[\mu_n(\tau)-\mu_n(0)]\dot{\mu}_n(t)$, we have
\begin{equation}
\begin{aligned}[b]
\|\varrho_0-\varrho_\tau\|_1&=\int_0^\tau\Tr{\tilde{X}_t\dot{\tilde{\varrho}}_t}dt\\
&=\int_0^\tau\avg{\tilde{X}_t,\mca{O}_{\tilde{\varrho}_t}(t,-\ln\tilde{\varrho}_t+\ln\tilde{\pi}_t)}dt\\
&\le\bra{\int_0^\tau\avg{\tilde{X}_t,\mca{O}_{\tilde{\varrho}_t}(t,\tilde{X}_t)}dt}^{1/2}\bra{\int_0^\tau\avg{-\ln\tilde{\varrho}_t+\ln\tilde{\pi}_t,\mca{O}_{\tilde{\varrho}_t}(t,-\ln\tilde{\varrho}_t+\ln\tilde{\pi}_t)}dt}^{1/2}\\
&=\bra{\int_0^\tau\avg{\tilde{X}_t,\mca{O}_{\tilde{\varrho}_t}(t,\tilde{X}_t)}dt}^{1/2}{\Sigma_\tau}^{1/2}.\label{eq:trace.bound.1}
\end{aligned}
\end{equation}
The first term in the last line of Eq.~\eqref{eq:trace.bound.1} can be rewritten as
\begin{equation}
\begin{aligned}[b]
\avg{\tilde{X}_t,\mca{O}_{\tilde{\varrho}_t}(t,\tilde{X}_t)}&=\frac{1}{2}\sum_{k}e^{-\beta\omega_k(t)/2}\avg{\tilde{X}_t,[\tilde{L}_k(t),\sop{\tilde{\varrho}_t}_{\beta\omega_k(t)}([\tilde{L}_k(t)^\dagger,\tilde{X}_t])]}\\
&=\frac{1}{2}\sum_{k}e^{-\beta\omega_k(t)/2}\Tr{[\tilde{X}_t^\dagger,\tilde{L}_k(t)]\sop{\tilde{\varrho}_t}_{\beta\omega_k(t)}([\tilde{L}_k(t)^\dagger,\tilde{X}_t])}\\
&=\frac{1}{2}\sum_{k}e^{-\beta\omega_k(t)/2}\avg{[\tilde{L}_k(t)^\dagger,\tilde{X}_t],\sop{\tilde{\varrho}_t}_{\beta\omega_k(t)}([\tilde{L}_k(t)^\dagger,\tilde{X}_t])}.
\end{aligned}
\end{equation}
Applying Propositions \ref{prop:sop.bound} and \ref{prop:mat.ine}, we have
\begin{equation}\label{eq:trace.bound.2}
\begin{aligned}[b]
&\int_0^\tau\avg{\tilde{X}_t,\mca{O}_{\tilde{\varrho}_t}(t,\tilde{X}_t)}dt\\
&=\int_0^\tau\frac{1}{2}\sum_{k}e^{-\beta\omega_k(t)/2}\avg{[\tilde{L}_k(t)^\dagger,\tilde{X}_t],\sop{\tilde{\varrho}_t}_{\beta\omega_k(t)}([\tilde{L}_k(t)^\dagger,\tilde{X}_t])}dt\\
&\le\int_0^\tau\frac{1}{4}\sum_{k}e^{-\beta\omega_k(t)/2}(e^{\beta\omega_k(t)/2}\Tr{\tilde{\varrho}_t[\tilde{L}_k(t)^\dagger,\tilde{X}_t][\tilde{L}_k(t)^\dagger,\tilde{X}_t]^\dagger}+e^{-\beta\omega_k(t)/2}\Tr{\tilde{\varrho}_t[\tilde{L}_k(t)^\dagger,\tilde{X}_t]^\dagger[\tilde{L}_k(t)^\dagger,\tilde{X}_t]})dt\\
&\le\int_0^\tau\sum_{k}\bras{\Tr{\tilde{\varrho}_t\tilde{L}_k(t)^\dagger\tilde{L}_k(t)}+e^{-\beta\omega_k(t)}\Tr{\tilde{\varrho}_t\tilde{L}_k(t)\tilde{L}_k(t)^\dagger}}dt\\
&=2\int_0^\tau\sum_{k}\Tr{\tilde{\varrho}_t\tilde{L}_k(t)^\dagger\tilde{L}_k(t)}dt\\
&=2\int_0^\tau\sum_{k}\Tr{\varrho_tL_k(t)^\dagger L_k(t)}dt.
\end{aligned}
\end{equation}
From Eqs.~\eqref{eq:trace.bound.1} and \eqref{eq:trace.bound.2}, we obtain a lower bound on the entropy production as
\begin{equation}
\Sigma_\tau\ge\frac{\|{\varrho}_0-{\varrho}_\tau\|_1^2}{2\tau\overline{\gamma}_\tau},\label{eq:trace.bound.3}
\end{equation}
where $\overline{\gamma}_\tau\coloneqq\tau^{-1}\int_0^\tau\sum_{k}\Tr{\varrho_tL_k(t)^\dagger L_k(t)}dt$ is the time average number of jumps and also known as the dynamical activity of the system \cite{Garrahan.2017.PRE,Shiraishi.2018.PRL,Vu.2019.PRE.UnderdampedTUR,Maes.2020.PR}.
Equation~\eqref{eq:trace.bound.3} reduces to the first inequality in Ineq.~(\DistanceBound) in the main text,
\begin{equation}
\beta Q\ge \Delta S+\frac{\|\varrho_0-\varrho_\tau\|_1^2}{2\tau\overline{\gamma}_\tau}.\label{eq:trace.bound.4}
\end{equation}
It is worth noting that the coefficient $\overline{\gamma}_\tau$ can be further bounded from above by the relaxation rates as
\begin{equation}
\overline{\gamma}_\tau\le\tau^{-1}\int_0^\tau\sum_{k}\|L_k(t)\|_\infty^2dt.
\end{equation}
Finally, we further bound the dissipated heat in terms of the reset error $\epsilon=\|\varrho_\tau-\Mat{0}\|_1$.
Applying the triangle inequality for the trace norm, we have
\begin{equation}
\begin{aligned}[b]
\|\varrho_0-\varrho_\tau\|_1&\ge\max\{\|\varrho_0-\Mat{0}\|_1-\|\varrho_\tau-\Mat{0}\|_1,\|\Mat{0}-\varrho_\tau\|_1-\|\Mat{0}-\varrho_0\|_1\}\\
&=|\|\varrho_0-\Mat{0}\|_1-\|\varrho_\tau-\Mat{0}\|_1|\\
&=|2(1-1/d)-\epsilon|.
\end{aligned}
\end{equation}
Here, we have used $\|\varrho_0-\Mat{0}\|_1=2(1-1/d)$ in the last line.
Consequently, we obtain the lower bound in terms of the reset error $\epsilon$ as
\begin{equation}
\beta Q\ge \Delta S+\frac{\|\varrho_0-\varrho_\tau\|_1^2}{2\tau\overline{\gamma}_\tau}\ge \Delta S+\frac{[2(1-1/d)-\epsilon]^2}{2\tau\overline{\gamma}_\tau}.\label{eq:trace.bound.5}
\end{equation}

\subsection{Behavior of $\overline{\gamma}_\tau$ in the long-time limit}
Here we show that the time-averaged dynamical activity $\overline{\gamma}_\tau$ becomes time-independent in the long-time limit $\tau\to\infty$.
In other words, we prove that $\overline{\gamma}_\tau$ converges to a constant as the duration of the protocol $\tau$ increases.
For convenience, we denote the Hamiltonian and jump operators by $H(\lambda_t)$ and $L_k(\lambda_t)$ to explicitly express the dependence on the control protocol $\lambda_t$.
The control protocol $\lambda_t$ can be parameterized as $\lambda_t=\hat{\lambda}_{t/\tau}$, where $\{\hat{\lambda}_t\}_{0\le t\le 1}$ is a fixed reference protocol.
It is convenient to define time-rescaled quantities
\begin{equation}
\hat{\varrho}_t=\varrho_{\tau t},~\hat{\mca{L}}_t=\mca{L}_{\tau t}.
\end{equation}
The Lindblad master equation can be rewritten as $\dot{\hat{\varrho}}_t=\tau\hat{\mca{L}}_t$.
It should be noted that the super-operator $\hat{\mca{L}}_t$ is independent of $\tau$, while the time-rescaled solution $\hat{\varrho}_t$ is not.
As $\tau\to\infty$, the density matrix of the system remains close to the instantaneous equilibrium state.
Therefore, $\hat{\varrho}_t$ can be expanded in terms of $1/\tau$ as \cite{Cavina.2017.PRL}
\begin{equation}
\hat\varrho_t=\hat\pi_t+\frac{1}{\tau}\hat\varrho_t^{(1)}+\frac{1}{\tau^2}\hat\varrho_t^{(2)}+\dots=\hat\pi_t+\tau^{-1}\mca{O}(1).
\end{equation}
Here, $\hat\pi_t=e^{-\beta H(\hat{\lambda}_t)}/\Tr{e^{-\beta H(\hat{\lambda}_t)}}$ is the instantaneous equilibrium state that is independent of $\tau$, and $\hat{\varrho}_t^{(n)}$ are traceless operators.
Using this perturbation expansion, we can calculate $\overline{\gamma}_\tau$ as follows:
\begin{equation}\label{eq:gamma.inf}
\begin{aligned}[b]
\overline{\gamma}_\tau&=\tau^{-1}\int_0^\tau\sum_k\Tr{L_k(\lambda_t)\varrho_tL_k(\lambda_t)^\dagger}dt\\
&=\int_0^1\sum_k\Tr{L_k(\hat{\lambda}_t)\hat{\varrho}_tL_k(\hat{\lambda}_t)^\dagger}dt\\
&=\int_0^1\sum_k\Tr{L_k(\hat{\lambda}_t)\hat{\pi}_tL_k(\hat{\lambda}_t)^\dagger}dt+\tau^{-1}\mca{O}(1)\\
&\eqqcolon\overline{\gamma}_\infty+\tau^{-1}\mca{O}(1).
\end{aligned}
\end{equation}
As $\tau\to\infty$, the second term in Eq.~\eqref{eq:gamma.inf} vanishes, and $\overline{\gamma}_\tau$ thus converges to the first term, $\overline{\gamma}_\infty$, which is a constant independent of $\tau$.

\section{Lower bound in terms of coherence [Ineq.~(\CoherenceBound) in the main text]}
Here we present the derivation of the lower bound in terms of quantum coherence in the main text.

\subsection{Derivation for the $d=2$ case}
To best convey the idea of the derivation, we first present the proof for the $d=2$ case.
The proof for the generic case will be presented later in an analogous way.
A two-level system can be described with two jump operators $L_1(t)=\sqrt{\gamma_1(t)}\dMat{0_t}{1_t}$ and $L_2(t)=\sqrt{\gamma_2(t)}\dMat{1_t}{0_t}$, where $\gamma_i(t)>0$ denotes the decay rate.
Let $\varrho_t=\sum_{m,n}c_{mn}(t)\dMat{m_t}{n_t}$, where $\{\Kt{n_t}\}$ is the energy eigenbasis of the Hamiltonian $H_t$ with the corresponding energy levels $\{\epsilon_n(t)\}$ (i.e., $H_t\Kt{n_t}=\epsilon_n(t)\Kt{n_t}$).
The corresponding density matrix in the interaction picture can be written as $\tilde{\varrho}_t=\sum_{m,n}c_{mn}(t)\dMat{\tilde{m}_t}{\tilde{n}_t}$, where $\Kt{\tilde{n}_t}\coloneqq U_t^\dagger\Kt{n_t}$.
The following term quantifies the quantum coherence generated during the erasure process:
\begin{equation}
\msc{C}_\tau\coloneqq\int_0^\tau C_{\ell_1}(\varrho_t)dt,
\end{equation}
where $C_{\ell_1}(\varrho_t)=\sum_{m\neq n}|c_{mn}(t)|$ is the $\ell_1$-norm coherence with respect to the energy eigenbasis of $H_t$.

Now, define $\tilde{X}_{01}(t)\coloneqq\dMat{\tilde{0}_t}{\tilde{1}_t}$ and $\tilde{X}_{10}(t)\coloneqq\dMat{\tilde{1}_t}{\tilde{0}_t}$, we can calculate
\begin{equation}\label{eq:2d.tmp1}
\begin{aligned}[b]
\Tr{\tilde{X}_{01}(t)^\dagger\dot{\tilde{\varrho}}_t}&=\davgObs{\tilde{0}_t}{\dot{\tilde{\varrho}}_t}{\tilde{1}_t}\\
&=\Br{\tilde{0}_t}\sum_{k}\bra{\tilde{L}_k(t)\tilde{\varrho}_t\tilde{L}_k(t)^\dagger-\frac{1}{2}\{\tilde{L}_k(t)^\dagger\tilde{L}_k(t),\tilde{\varrho}_t\}}\Kt{\tilde{1}_t}\\
&=-\frac{\gamma_1(t)+\gamma_2(t)}{2}c_{01}(t)=-\frac{\gamma_t}{2}c_{01}(t),
\end{aligned}
\end{equation}
where $\gamma_t\coloneqq\gamma_1(t)+\gamma_2(t)=\sum_k\|L_k(t)\|_\infty^2$.
Analogously, we also have $\Tr{\tilde{X}_{10}(t)^\dagger\dot{\tilde{\varrho}}_t}=-\gamma_tc_{10}(t)/2$.
Define $\tilde{X}_t\coloneqq \gamma_t^{-1}\brab{f[c_{01}(t)]^*\tilde{X}_{01}(t)+f[c_{10}(t)]^*\tilde{X}_{10}(t)}$, which is a self-adjoint operator, we obtain the relation
\begin{equation}
|\Tr{\tilde{X}_t^\dagger\dot{\tilde{\varrho}}_t}|=\frac{1}{2}\bras{|c_{01}(t)|+|c_{10}(t)|}=\frac{1}{2}C_{\ell_1}(\varrho_t).
\end{equation}
Here, $f(x)=x^*/|x|$ if $x\neq 0$ and $f(x)=1$ otherwise.
It can also be easily calculated that $\|[\tilde{L}_k(t),\tilde{X}_t]\|_\infty=\gamma_t^{-1}\|\tilde{L}_k(t)\|_\infty$.
Applying the Cauchy--Schwarz inequality (proved in Proposition \ref{prop:cs.ine}), we have
\begin{equation}\label{eq:2d.tmp2}
\begin{aligned}[b]
\frac{1}{2}\msc{C}_\tau&=\frac{1}{2}\int_0^\tau C_{\ell_1}(\varrho_t)dt\\
&=\int_0^\tau|\Tr{\tilde{X}_t^\dagger\dot{\tilde{\varrho}}_t}|dt\\
&=\int_0^\tau\left|\avg{\tilde{X}_t,\mca{O}_{\tilde{\varrho}_t}(t,-\ln\tilde{\varrho}_t+\ln\tilde{\pi}_t)}\right|dt\\
&\le\bra{\int_0^\tau\avg{\tilde{X}_t,\mca{O}_{\tilde{\varrho}_t}(t,\tilde{X}_t)}dt}^{1/2}\bra{\int_0^\tau\avg{-\ln\tilde{\varrho}_t+\ln\tilde{\pi}_t,\mca{O}_{\tilde{\varrho}_t}(t,-\ln\tilde{\varrho}_t+\ln\tilde{\pi}_t)}dt}^{1/2}\\
&=\bra{\int_0^\tau\avg{\tilde{X}_t,\mca{O}_{\tilde{\varrho}_t}(t,\tilde{X}_t)}dt}^{1/2}{\Sigma_\tau}^{1/2}.
\end{aligned}
\end{equation}
The first term in the last line of Eq.~\eqref{eq:2d.tmp2} can be rewritten as
\begin{equation}\label{eq:2d.tmp3}
\begin{aligned}[b]
\avg{\tilde{X}_t,\mca{O}_{\tilde{\varrho}_t}(t,\tilde{X}_t)}&=\frac{1}{2}\sum_{k}e^{-\beta\omega_k(t)/2}\avg{[\tilde{L}_k(t)^\dagger,\tilde{X}_t],\sop{\tilde{\varrho}_t}_{\beta\omega_k(t)}([\tilde{L}_k(t)^\dagger,\tilde{X}_t])}.
\end{aligned}
\end{equation}
Applying Proposition \ref{prop:sop.bound} with $Y=[\tilde{L}_k(t)^\dagger,\tilde{X}_t]$ and $\theta=\beta\omega_k(t)$, we obtain
\begin{equation}\label{eq:2d.tmp4}
\avg{[\tilde{L}_k(t)^\dagger,\tilde{X}_t],\sop{\tilde{\varrho}_t}_{\beta\omega_k(t)}([\tilde{L}_k(t)^\dagger,\tilde{X}_t])}\le\frac{1}{2}(e^{-\beta\omega_k(t)/2}+e^{\beta\omega_k(t)/2})\|[\tilde{L}_k(t)^\dagger,\tilde{X}_t]\|_\infty^2.
\end{equation}
Inserting Eq.~\eqref{eq:2d.tmp4} into Eq.~\eqref{eq:2d.tmp3}, we then have
\begin{equation}
\avg{\tilde{X}_t,\mca{O}_{\tilde{\varrho}_t}(t,\tilde{X}_t)}\le\frac{1}{4}\sum_{k}e^{-\beta\omega_k(t)/2}(e^{-\beta\omega_k(t)/2}+e^{\beta\omega_k(t)/2})\|[\tilde{L}_k(t)^\dagger,\tilde{X}_t]\|_\infty^2=\frac{1}{2}\sum_{k}\|[\tilde{L}_k(t),\tilde{X}_t]\|_\infty^2,\label{eq:2d.tmp5}
\end{equation}
where we have used $\tilde{L}_k(t)^\dagger=e^{\beta\omega_k(t)/2}\tilde{L}_{k'}(t)$, $\tilde{X}_t$ is self-adjoint, and $\|A^\dagger\|_\infty=\|A\|_\infty$ for any operator $A$.
From Eqs.~\eqref{eq:2d.tmp2} and \eqref{eq:2d.tmp5}, we obtain
\begin{equation}
\frac{1}{2}\msc{C}_\tau\le\bra{\frac{1}{2}\int_0^\tau\sum_{k}\|[\tilde{L}_k(t),\tilde{X}_t]\|_\infty^2dt}^{1/2}{\Sigma_\tau}^{1/2}\Rightarrow\Sigma_\tau\ge\frac{\overline{\gamma}_\tau^{\rm R}\msc{C}_\tau^2}{2\tau},
\end{equation}
from which Ineq.~(\CoherenceBound) in the main text is immediately derived,
\begin{equation}
\beta Q \ge \Delta S+\frac{\overline{\gamma}_\tau^{\rm R}\msc{C}_\tau^2}{2\tau}.
\end{equation}
Here, the system-dependent quantity $\overline{\gamma}_\tau^{\rm R}$ is defined via the relation
\begin{equation}
(\overline{\gamma}_\tau^{\rm R})^{-1}\coloneqq\tau^{-1}\int_0^\tau\sum_{k}\|[\tilde{L}_k(t),\tilde{X}_t]\|_\infty^2dt=\tau^{-1}\int_0^\tau(\sum_{k}\|L_k(t)\|_\infty^2)^{-1}dt=\tau^{-1}\int_0^\tau\gamma_t^{-1}dt.
\end{equation}

\subsection{Derivation for the generic case}
We consider the generic case in which each jump operator could describe multiple transitions between energy eigenstates with the same energy change.
For simplicity, we assume that the Hamiltonian's energy levels are non-degenerate, $H_t=\sum_n\epsilon_n(t)\Mat{n_t}$ with $\epsilon_n(t)\neq\epsilon_m(t)$ for all $m\neq n$.
In the case of energy degeneracy, essentially the same result can be obtained by using quantum coherence between energy eigenstates with different energies.
Note that the set of all ordered pairs $(m,n)$ with $m\neq n$ can be divided into several subsets $\{{S}_i\}$ such that all pairs $(m,n)$ in a subset have the same value of the energy change $\epsilon_m(t)-\epsilon_n(t)$.
Then, each jump operator $L_{k}(t)$ can be written in the general form
\begin{equation}
L_{k}(t)=\sum_{(m,n)\in{S}_{i_k}}\Mat{n_t}L_k\Mat{m_t},
\end{equation}
where $i_k$ is some index and $L_k$ is an operator.
It is straightforward to show that the jump operators in the interaction picture read
\begin{equation*}
\tilde{L}_{k}(t)=\sum_{(m,n)\in{S}_{i_k}}\davgObs{n_t}{L_k}{m_t}\dMat{\tilde{n}_t}{\tilde{m}_t}\quad\textrm{and}\quad\tilde{L}_{k}(t)^\dagger\tilde{L}_{k}(t)=\sum_{(m,n)\in{S}_{i_k}}|\davgObs{n_t}{L_k}{m_t}|^2\Mat{\tilde{m}_t}.
\end{equation*}
Define $\tilde{X}_{pq}(t)\coloneqq\dMat{\tilde{p}_t}{\tilde{q}_t}$. Then, for any subset $S_i$ and any pair $(p,q)\in{S}_i$, we have
\begin{equation}\label{eq:tmp.gcc}
\begin{aligned}[b]
&\Tr{\tilde{X}_{pq}(t)^\dagger\dot{\tilde{\varrho}}_t}\\
&=\davgObs{\tilde{p}_t}{\dot{\tilde{\varrho}}_t}{\tilde{q}_t}\\
&=\sum_{k}\Br{\tilde{p}_t}\bra{\tilde{L}_{k}(t)\tilde{\varrho}_t\tilde{L}_{k}(t)^\dagger-\frac{1}{2}\{\tilde{L}_{k}(t)^\dagger\tilde{L}_{k}(t),\tilde{\varrho}_t\}}\Kt{\tilde{q}_t}\\
&=\sum_{k}\Big[{\sum_{(m,n),(m',n')\in{S}_{i_k}}\delta_{pn}\delta_{qn'}\davgObs{n_t}{L_k}{m_t}\davgObs{m_t'}{L_k^\dagger}{n_t'}\davgObs{\tilde{m}_t}{\tilde{\varrho}_t}{\tilde{m}_t'}-\frac{1}{2}\sum_{(m,n)\in{S}_{i_k}}(\delta_{pm}+\delta_{qm})|\davgObs{n_t}{L_k}{m_t}|^2\davgObs{\tilde{p}_t}{\tilde{\varrho}_t}{\tilde{q}_t}}\Big]\\
&=-\frac{1}{2}\sum_{(m,m')\in{S}_i}\gamma_{mm'}^{pq}(t)c_{mm'}(t),
\end{aligned}
\end{equation}
where we have used $\epsilon_{m}(t)-\epsilon_{n}(t)=\epsilon_{m'}(t)-\epsilon_{n'}(t)\Rightarrow\epsilon_{m}(t)-\epsilon_{m'}(t)=\epsilon_{n}(t)-\epsilon_{n'}(t)=\epsilon_{p}(t)-\epsilon_{q}(t)$, given that $\delta_{pn}$ and $\delta_{qn'}$ are nonzero.
The coefficients $\gamma_{mm'}^{pq}(t)$ are defined as
\begin{align}
\gamma_{mm'}^{pq}(t)&\coloneqq -2\sum_{k}\sum_{(m,n),(m',n')\in{S}_{i_k}}\delta_{pn}\delta_{qn'}\davgObs{n_t}{L_k}{m_t}\davgObs{m_t'}{L_k^\dagger}{n_t'},~\textrm{for}~(m,m')\in S_i\setminus(p,q),\\
\gamma_{pq}^{pq}(t)&\coloneqq\sum_k\sum_{(m,n)\in{S}_{i_k}}(\delta_{pm}+\delta_{qm})|\davgObs{n_t}{L_k}{m_t}|^2>0.
\end{align}
Notably, we can prove that the matrix $[\gamma_{mm'}^{pq}(t)]$ is invertible; the proof is presented in Proposition \ref{prop:nsin}.
Thus, there exist coefficients $z_{pq}(t)\in\mbb{C}$ such that
\begin{equation}
\sum_{(p,q)\in{S}_i}z_{pq}(t)\gamma_{mm'}^{pq}(t)=f[c_{mm'}(t)]~\forall (m,m')\in{S}_i,
\end{equation}
where the function $f(\cdot)$ is defined as in the $d=2$ case.
Consequently, we have
\begin{equation}
\sum_{(p,q)\in S_i}z_{pq}(t)\Tr{\tilde{X}_{pq}(t)^\dagger\dot{\tilde{\varrho}}_t}=-\frac{1}{2}\sum_{(m,m')\in S_i}|c_{mm'}(t)|.\label{eq:tmp.const.Xpq}
\end{equation}
Defining the self-adjoint operator $\tilde{X}_t\coloneqq\sum_{i}\sum_{(p,q)\in S_i}z_{pq}(t)^*\tilde{X}_{pq}(t)$ and taking the sum of both sides of Eq.~\eqref{eq:tmp.const.Xpq} over all $i$, we obtain
\begin{equation}
|\Tr{\tilde{X}_t^\dagger\dot{\tilde{\varrho}}_t}|=\frac{1}{2}\sum_{m\neq n}|c_{mn}(t)|=\frac{1}{2}C_{\ell_1}(\varrho_t).
\end{equation}
Following the same procedure as in the $d=2$ case, we obtain Ineq.~(\CoherenceBound) in the main text,
\begin{equation}
\beta Q \ge \Delta S+\frac{\overline{\gamma}_\tau^{\rm R}\msc{C}_\tau^2}{2\tau},
\end{equation}
where $(\overline{\gamma}_\tau^{\rm R})^{-1}\coloneqq\tau^{-1}\int_0^\tau\sum_{k}\|[\tilde{L}_k(t),\tilde{X}_t]\|_\infty^2dt$.
Generally, it is difficult to obtain a more analytical form of $\overline{\gamma}_\tau^{\rm R}$, except when each jump operator characterizes a single transition between two energy levels, i.e., when $L_k(t)=\sqrt{\gamma_k(t)}\dMat{k_t^+}{k_t^-}$ for all $k$; in this case, $\overline{\gamma}_\tau^{\rm R}$ can be written solely in terms of $\gamma_k(t)$.
Nevertheless, we can further bound $\overline{\gamma}_\tau^{\rm R}$ from below in terms of $\davgObs{n_t}{L_k}{m_t}$, which, however, results in a looser bound.
Using the inequalities $\|[A,B]\|_\infty\le 2\|AB\|_\infty$ and $\|AB\|_\infty\le\|A\|_\infty\|B\|_\infty$ for arbitrary operators $A$ and $B$ and noting that $\|\tilde{L}_k(t)^\dagger\|_\infty=\|L_k(t)\|_\infty$, we have
\begin{equation}
\|[\tilde{L}_k(t),\tilde{X}_t]\|_\infty^2\le 4\|L_k(t)\|_\infty^2\|\tilde{X}_t\|_\infty^2.
\end{equation}
Moreover, since $\|A\|_\infty\le\|A\|_{\rm F}\coloneqq\sqrt{\Tr{A^\dagger A}}$ and $\|A\mbm{x}\|_{\rm F}\le\|A\|_\infty\|\mbm{x}\|_{\rm F}$, we can bound $\|\tilde{X}_t\|_\infty^2$ from above as
\begin{equation}
\|\tilde{X}_t\|_\infty^2\le \|\tilde{X}_t\|_{\rm F}^2=\|\Xi_t^{-1}\mbm{f}_t\|_{\rm F}^2\le\|\Xi_t^{-1}\|_\infty^2\|\mbm{f}_t\|_{\rm F}^2,
\end{equation}
where $\Xi_t\coloneqq{\rm diag}([\gamma_{mm'}^{pq}(t)]_{S_1},\dots,[\gamma_{mm'}^{pq}(t)]_{S_N})$ is a block diagonal matrix, $\mbm{f}_t\coloneqq [[f(c_{mm'}(t))]_{S_1},\dots,[f(c_{mm'}(t))]_{S_N}]^\top$ is a vector, and $N$ is the number of subsets $S_i$.
Note that $\|\mbm{f}_t\|_{\rm F}^2=\sum_{m\neq m'}|f(c_{mm'}(t))|^2=d(d-1)$ and $\|\Xi_t^{-1}\|_\infty^2=\max_{i}\|[\gamma_{mm'}^{pq}(t)]_{S_i}^{-1}\|_\infty^2$, we arrive at a lower bound of $\overline{\gamma}_\tau^{\rm R}$ as follows:
\begin{equation}
\overline{\gamma}_\tau^{\rm R}\ge \frac{\tau}{4d(d-1)\int_0^\tau\gamma_t\max_{i}\|[\gamma_{mm'}^{pq}(t)]_{S_i}^{-1}\|_\infty^2dt}.
\end{equation}

\begin{proposition}\label{prop:nsin}
The matrix $\Gamma=[\gamma_{mm'}^{pq}(t)]_{S_i}$ with $(p,q),(m,m')\in S_i$ is invertible. Here, the subscript $(mm')$ and superscript $(pq)$ indicate row and column indices, respectively.
\end{proposition}
\begin{proof}
First, let us state some definitions and existing results.
A matrix $A=[a_{ij}]$ is said to be \emph{weakly} diagonally dominant if $|a_{ii}|\ge\sum_{j(\neq i)}|a_{ij}|$ for every row $i$ of the matrix, where $a_{ij}$ denotes the entry in the $i$th row and $j$th column.
We also say the row $i$ of the matrix is \emph{strictly} diagonally dominant if the strict inequality holds, i.e., $|a_{ii}|>\sum_{j(\neq i)}|a_{ij}|$.
It was proved \cite{Shivakumar.1974.PAMS} that if (i) the matrix $A$ is weakly diagonally dominant and (ii) for each row $i$ there is a sequence of nonzero elements of $A$ of the form $a_{ii_1},a_{i_1i_2},\dots,a_{i_rj}$ with $j$ the index of some strictly diagonally dominant row, then the matrix $A$ is invertible.

It can be shown that the matrix $\Gamma$ is a weakly diagonally dominant matrix, i.e., the following inequality holds for all $(m,m')\in S_i$,
\begin{equation}
|\gamma_{mm'}^{mm'}(t)|\ge\sum_{(p,q)\in S_i\setminus(m,m')}|\gamma_{mm'}^{pq}(t)|.
\end{equation}
This can be proved as follows:
\begin{equation}
\begin{aligned}[b]
\sum_{(p,q)\in S_i\setminus(m,m')}|\gamma_{mm'}^{pq}(t)|&=2\sum_{(p,q)\in S_i\setminus(m,m')}|\sum_{k}\sum_{(m,n),(m',n')\in{S}_{i_k}}\delta_{pn}\delta_{qn'}\davgObs{n_t}{L_k}{m_t}\davgObs{m_t'}{L_k^\dagger}{n_t'}|\\
&\le 2\sum_{(p,q)\in S_i\setminus(m,m')}\sum_{k}\sum_{(m,n),(m',n')\in{S}_{i_k}}\delta_{pn}\delta_{qn'}|\davgObs{n_t}{L_k}{m_t}\davgObs{m_t'}{L_k^\dagger}{n_t'}|\\
&\le \sum_{(p,q)\in S_i\setminus(m,m')}\sum_{k}\sum_{(m,n),(m',n')\in{S}_{i_k}}\bra{\delta_{pn}\delta_{qn'}|\davgObs{n_t}{L_k}{m_t}|^2+\delta_{pn}\delta_{qn'}|\davgObs{m_t'}{L_k^\dagger}{n_t'}|^2}\\
&\le \sum_{k}\big[\sum_{(m,n)\in{S}_{i_k}}{|\davgObs{n_t}{L_k}{m_t}|^2+\sum_{(m',n')\in{S}_{i_k}}|\davgObs{m_t'}{L_k^\dagger}{n_t'}|^2}\big]\\
&=\sum_k\sum_{(p,q)\in{S}_{i_k}}(\delta_{mp}|\davgObs{q_t}{L_k}{p_t}|^2+\delta_{m'p}|\davgObs{q_t}{L_k}{p_t}|^2)\\
&=|\gamma_{mm'}^{mm'}(t)|.
\end{aligned}
\end{equation}
Moreover, since we consider the case where the energy eigenstates are irreducible (i.e., the system state can change from an instantaneous energy eigenstate to any other instantaneous eigenstate via a finite number of jumps), the matrix $\Gamma$ has at least one strictly diagonally dominant row and satisfies the property (ii).
Consequently, the matrix $\Gamma$ is invertible.
\end{proof}

\section{Numerical illustration with the optimal control protocol}
The density matrix can be represented using the Bloch vector, $\varrho_t=(\mbb{I}+\mbm{\sigma}\cdot\mbm{r}_t)/2$, where $\mbm{r}_t\coloneqq [r_x(t),r_y(t),r_z(t)]^\top$ is the Bloch vector and $\mbm{\sigma}\coloneqq[\sigma_x,\sigma_y,\sigma_z]^\top$ is the vector of Pauli matrices.
Note that $\|\mbm{r}_t\|^2\coloneqq r_x(t)^2+r_y(t)^2+r_z(t)^2\le 1$.
The Lindblad master equation can be translated into a differential equation describing the time evolution of the vector $\mbm{r}_t$ as
\begin{equation}
\dot{\mbm{r}}_t=\Upsilon_t\mbm{r}_t+\mbm{b}_t,\label{eq:rt.evol}
\end{equation}
where $\Upsilon_t=[\upsilon_{mn}(t)]\in\mbb{R}^{3\times 3}$ is a real matrix with
\begin{equation}
\upsilon_{mn}(t)=\frac{1}{2}\Tr{i[\sigma_m,\sigma_n]H_t+\sum_k\bras{L_k(t)\sigma_nL_k(t)^\dagger\sigma_m-\frac{1}{2}(L_k(t)^\dagger L_k(t)\sigma_n\sigma_m+L_k(t)^\dagger L_k(t)\sigma_m\sigma_n)}},
\end{equation}
and $\mbm{b}_t\coloneqq [b_x(t),b_y(t),b_z(t)]^\top$ is a real vector with $b_n(t)=\sum_k\Tr{[L_k(t),L_k(t)^\dagger]\sigma_n}/2$.

Since $F(\varrho_\tau,\Mat{0})=(1-r_z(\tau))/2$, the multi-objective functional in the optimization problem can be rewritten in terms of $\mbm{r}_t$ as
\begin{equation}
\mca{F}[\{\epsilon_t,\theta_t\}]=-\frac{1}{2}\lambda\tau\beta\int_0^\tau\Tr{(\mbm{\sigma}\cdot\dot{\mbm{r}}_t)H_t}dt-(1-\lambda)\frac{1-r_z(\tau)}{2}.
\end{equation}
The functional $\mca{F}$ has two objectives, a path cost (dissipated heat) and a terminal cost (quantum fidelity). 
These two objectives are incompatible, i.e., they cannot be simultaneously optimal.
Increasing the quantum fidelity has to pay a price in dissipation.
Conversely, reducing only dissipation could degrade the quantum fidelity between the final state and the ground state.
Therefore, optimizing the functional $\mca{F}$ for all $\lambda\in[0,1)$ is equivalent to finding the boundary of the space of all feasible protocols, where heat dissipation cannot be minimized further without degrading the quantum fidelity, or vice versa.
The protocols forming the boundary are known as the Pareto-optimal protocols, while the collection of the Pareto-optimal protocols is called the Pareto front.

The functional $\mca{F}$ is optimized under the equality constraint Eq.~\eqref{eq:rt.evol} and the inequality constraints $0.4\le\epsilon_t\le 10$ and $-\pi\le\theta_t\le\pi$.
There are two classes of methods for numerically solving the optimal protocol, indirect and direct methods.
The indirect methods are based on the calculus of variations and are difficult to deal with inequality constraints.
Therefore, we employ the latter, which transcribes the infinite-dimensional problem into a finite-dimensional, nonlinear programming problem.
Specifically, we discretize the control parameters into $N=400$ points and optimize the functional $\mca{F}$ with the help of nonlinear programming solvers \cite{Koenemann.2019}.
The discontinuity of the control parameters at times $t=0$ and $t=\tau$ is allowed, which is known as a generic feature of optimal protocols \cite{Schmiedl.2007.PRL,Then.2008.PRE,Aurell.2011.PRL,Solon.2018.PRL}.
It is worth noting that the condition for the validity of the rotating-wave approximation (i.e., $\chi_t\ll\epsilon_t$) is always satisfied for the parameters of the optimal and non-optimal controls in this work.
Here, $\chi_t\coloneqq\alpha\epsilon_t\coth(\beta\epsilon_t/2)/2$ denotes the temporal average relaxation rate.

We vary the value of $\lambda$ and obtain the Pareto front for the $\tau=5$ and $\tau=10$ cases.
The results are plotted in Fig.~\ref{fig:result0}.
As can be seen, there exists a tradeoff between dissipation and erasure fidelity.
To increase the erasure fidelity (i.e., reducing $1-F$), one has to increase heat dissipation.
The area below the Pareto front implies infeasible protocols and characterizes the values of $Q$ and $1-F$ that cannot be achieved simultaneously.
The results are also consistent with our intuition that the longer the operational time, the larger the space of feasible protocols.

\begin{figure}[!h]
\centering
\includegraphics[width=0.6\linewidth]{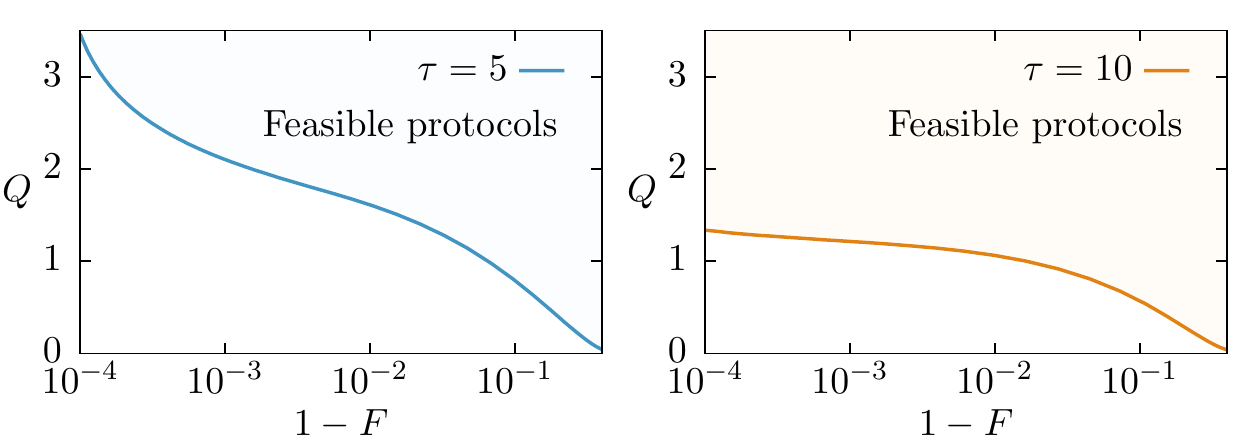}
\protect\caption{The Pareto fronts for different operational times, $\tau=5$ (left panel) and $\tau=10$ (right panel). The filled regions depict feasible protocols, whereas the area below the Pareto front (the solid line) is inaccessible. The parameters are set to $\alpha=0.2$, $\beta=1$, $\epsilon_0=0.4$, and $\epsilon_\tau=10$.}\label{fig:result0}
\end{figure}

\subsection{Effect of quantum coherence on heat dissipation when $\varrho_0\neq\mbb{I}/2$}
To further understand the effect of quantum coherence on heat dissipation, we consider the case where information is described by the state $\varrho_0=(\mbb{I}-0.5\sigma_x-0.5\sigma_y+0.5\sigma_z)/2$.
It is worth emphasizing that resetting the qubit from this initial state to the ground state unavoidably generates coherence, even with the optimal protocol.
We numerically solve the optimal control protocol with $(1-\lambda)/\lambda=10^4$ and compare the results with those obtained by the non-optimal control protocol.
Note that both protocols drive the qubit to the ground state at the same order of error.
The time variation of the control parameters, the time evolution of the qubit, the amount of quantum coherence generated in the energy eigenbasis, and the dissipated heat associated with each protocol are plotted in Fig.~\ref{fig:result1}.
As can be seen in Fig.~\ref{fig:result1}(a), the control parameters in the two protocols vary in different ways.
In particular, the coherence parameter $\theta_t$ decreases from a value close to $\pi$ toward $0$ for the optimal protocol, whereas it gradually increases from $-\pi$ toward $0$ for the non-optimal protocol.
Consequently, the qubit evolves through different paths in each protocol, as shown in Fig.~\ref{fig:result1}(b).
Figure~\ref{fig:result1}(c) shows that the non-optimal protocol produces significantly more quantum coherence than the optimal protocol does.
As a result, the average dissipated heat of the optimal protocol is much smaller than that of the non-optimal protocol, as shown in Fig.~\ref{fig:result1}(d).
This numerical evidence positively supports the claim that the creation of quantum coherence should be avoided when erasing information.

\begin{figure}[!h]
\centering
\includegraphics[width=1.0\linewidth]{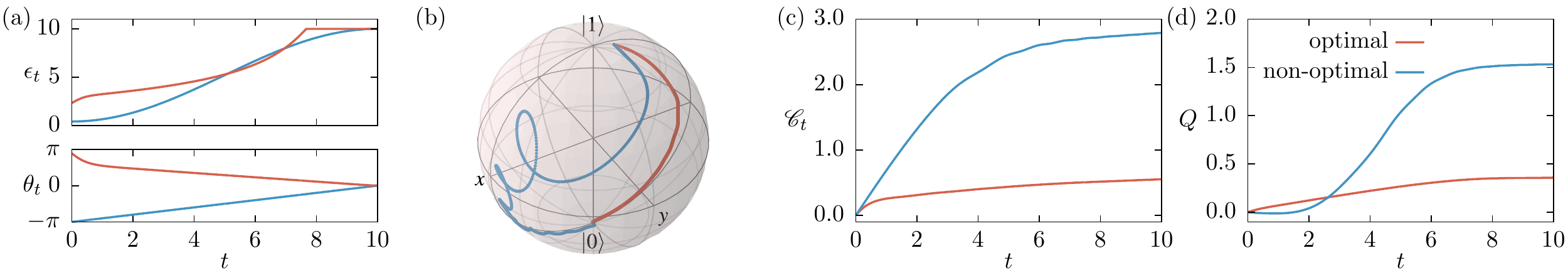}
\protect\caption{Comparison of the optimal and non-optimal control protocols. The numerical results obtained with the optimal and non-optimal protocols are depicted by the red and blue lines, respectively. The non-optimal protocol is given by $\epsilon_t=\epsilon_0+(\epsilon_\tau-\epsilon_0)\sin(\pi t/2\tau)^2$ and $\theta_t=\pi(t/\tau-1)$ \cite{Miller.2020.PRL.QLP}. (a) The time variation of control parameters. (b) A geometrical representation of the time evolution of the qubit. (c) The amount of quantum coherence produced in the energy eigenbasis at each instant of time. (d) The average dissipated heat over time for each protocol. The parameters are set to $\alpha=0.2$, $\beta=1$, $\epsilon_0=0.4$, $\epsilon_\tau=10$, and $\tau=10$.}\label{fig:result1}
\end{figure}

\subsection{Comparison of the bounds Ineqs.~(\DistanceBound) and (\CoherenceBound) when $\msc{C}_\tau\gg 1$}
Here we demonstrate the tightness of the bounds Ineqs.~(\DistanceBound) and (\CoherenceBound) in the presence of large quantum coherence.
As mentioned in the main text, the bound Ineq.~(\DistanceBound) is tighter than the bound Ineq.~(\CoherenceBound) when little or no coherence is produced.
In the following, we show that it is not the case when a large amount of quantum coherence is generated, i.e., the bound Ineq.~(\CoherenceBound) is typically stronger than the bound Ineq.~(\DistanceBound) when $\msc{C}_\tau\gg 1$.
To this end, we consider a non-optimal protocol that produces a large amount of quantum coherence, $\epsilon_t=\epsilon_0+(\epsilon_\tau-\epsilon_0)\sin(\pi t/2\tau)^2$ and $\theta_t=\pi\sin(20\pi t/\tau)^2$.
We vary the operational time $\tau$ and plot the amount of quantum coherence $\msc{C}_\tau$ and the finite-time correction terms, $\overline{\gamma}_\tau^{\rm R}\msc{C}_\tau^2/(2\tau)$ and $\|\varrho_0-\varrho_\tau\|_1^2/(2\tau\overline{\gamma}_\tau)$, as functions of $\tau$ in Fig.~\ref{fig:result2}.
Figure \ref{fig:result2}(a) shows the time variation of control parameters of the non-optimal protocol for the $\tau=50$ case.
As shown in Fig.~\ref{fig:result2}(b), $\msc{C}_\tau\gg 1$ for all $\tau$.
Simultaneously, the bound Ineq.~(\CoherenceBound) is stronger than the bound Ineq.~(\DistanceBound), which is illustrated in Fig.~\ref{fig:result2}(c).

\begin{figure}[!h]
\centering
\includegraphics[width=0.9\linewidth]{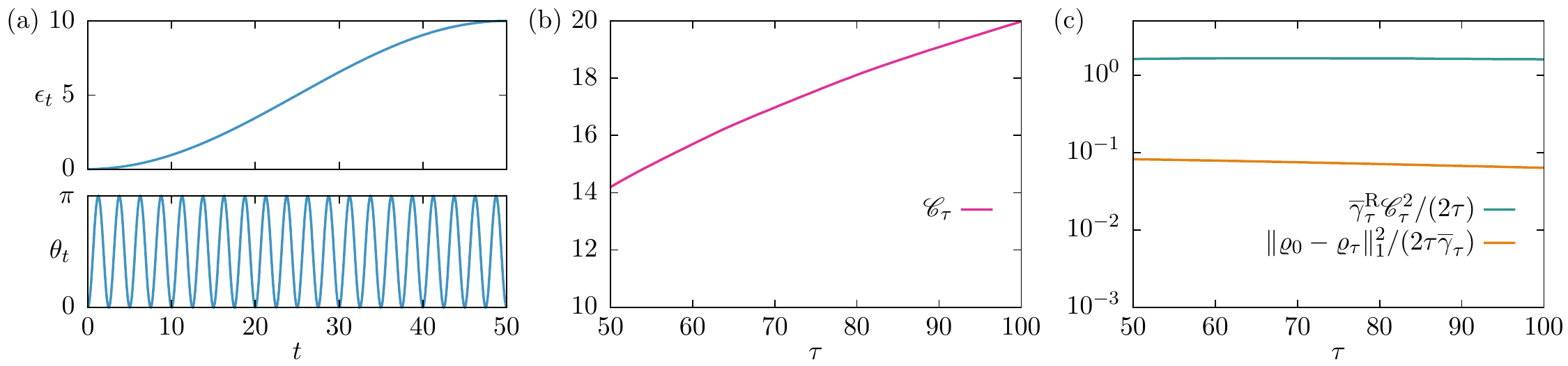}
\protect\caption{Comparison of the derived bounds Ineqs.~(\DistanceBound) and (\CoherenceBound) with a non-optimal control protocol given by $\epsilon_t=\epsilon_0+(\epsilon_\tau-\epsilon_0)\sin(\pi t/2\tau)^2$ and $\theta_t=\pi\sin(20\pi t/\tau)^2$. (a) The time variation of the control parameters for the $\tau=50$ case. (b) The amount of quantum coherence generated in the energy eigenbasis. (c) The correction terms of the derived bounds. The operational time $\tau$ is varied from $50$ to $100$, while the other parameters are set to $\alpha=0.2$, $\beta=1$, $\epsilon_0=0.4$, and $\epsilon_\tau=10$.}\label{fig:result2}
\end{figure}